%% file: main-full.tex
\documentclass[a4paper]{article}

\usepackage{microtype}
\usepackage{algorithm}
\usepackage{amssymb}
\usepackage[tbtags,fleqn]{amsmath}
\usepackage{enumerate}
\usepackage{amsthm}
\usepackage{array}
\usepackage{multirow}
\usepackage{tabularx}
\usepackage{color}
\usepackage[]{graphicx}
\usepackage[utf8]{inputenc}
\usepackage{amsmath}
\usepackage{comment}
\usepackage{subcaption}

\theoremstyle{plain}
\newtheorem{theorem}{Theorem}
\newtheorem{definition}[theorem]{Definition}
\newtheorem{lemma}[theorem]{Lemma}
\newtheorem{corollary}[theorem]{Corollary}

\title{Gathering on Rings for Myopic Asynchronous Robots with Lights\thanks{This work was supported in part by Hiroshima University, project ESTATE (Ref. ANR-16-CE25-0009-03), JSPS KAKENHI No. 17K00019, 18K11167 and 19K11828, and Israel \& Japan Science and Technology Agency (JST) SICORP (Grant\#JPMJSC1806).}}

\author{
Sayaka Kamei\thanks{Graduate School of Engineering, Hiroshima University, Japan}
\and
Anissa Lamani\thanks{Ecole internationale des sciences du traitement de l'information, Cergy, France.}
\and
Fukuhito Ooshita\thanks{Graduate School of Science and Technology, Nara Institute of Science and Technology, Japan}
\and
S\'ebastien Tixeuil\thanks{Sorbonne University, France}
\and 
Koichi Wada\thanks{Faculty of Science and Engineering, Hosei University, Japan}
}

\date{}

\begin{document}

\maketitle

\begin{abstract}
\input{abstract}
\end{abstract}

\section{Introduction}\label{sec:intro}
\input{introduction.tex}

\section{Model}\label{sec:model}
\input{model.tex}

\section{Algorithms}\label{sec:alg}
In this section, we propose two algorithms for myopic robots.
One is for the case that $M_{init}$ is odd, uses three colors ($K=3$).
The other is for the case that $M_{init}$ is even and $O_{init}$ is odd, uses four colors ($K=4$).
We assume that the initial configurations satisfy the following conditions:
\begin{itemize}
    \item All robots have the same color White,
    \item Each occupied node can have multiple robots, and
    \item $\phi\geq H_{init}\geq 1$ holds. 
\end{itemize} 

\subsection{Algorithm for the case $M_{init}$ is odd}

\subsubsection{Description}
\input{algorithm1-full.tex}

\subsubsection{Proof of correctness}
\input{alg1-correctness-full.tex}

\subsection{Algorithm for the case $M_{init}$ is even and $O_{init}$ is odd}

\subsubsection{Description}
\input{algorithm2-full.tex}

\subsubsection{Proof of correctness}
\input{alg2-correctness-full.tex}

\section{Discussion}
\label{sec:imposs}
In this section, we discuss the gathering problem in other cases.
\input{impossibility}

\section{Conclusion}\label{sec:conc}
We presented the first gathering algorithms for myopic luminous robots in rings. One algorithm considers the case where $M_{init}$ is odd, while the other is for the case where $O_{init}$ is odd. The hypotheses used for our algorithms closely follow the impossibility results found for the other cases. 

Some interesting questions remain open:
\begin{itemize}
    \item Are there any deterministic algorithms for the case where $M_{init}$ and $O_{init}$ are even (such solutions would have to avoid starting or ending up in an edge-view-symmetric situation)?
    \item Are there any algorithms for the case where $M_{init}$ (resp. $O_{init}$) is odd that use fewer colors than ours? The current lower bound for odd $M_{init}$ (resp. $O_{init}$) is $2$ (resp. $3$), but our solutions use $3$ (resp. $4$) colors.
    \item Are there any algorihtms for ring gathering that are \emph{not} cautious (a positive answer would enable starting configurations with a number of borders different from $2$)? 
\end{itemize}
\bibliographystyle{plainurl}
\bibliography{biblio2}

\end{document}

%% file: abstract.tex
We investigate gathering algorithms for asynchronous autonomous mobile robots moving in uniform ring-shaped networks. Different from most work using the Look-Compute-Move (LCM) model, we assume that robots have limited visibility and lights.
That is, robots can observe nodes only within a certain fixed distance, and emit a color from a set of constant number of colors.
We consider gathering algorithms depending on two parameters related to the initial configuration: $M_{init}$, which denotes the number of nodes between two border nodes, and $O_{init}$, which denotes the number of nodes hosting robots between two border nodes. In both cases, a border node is a node hosting one or more robots that cannot see other robots on at least one side. 
Our main contribution is to prove that, if $M_{init}$ or $O_{init}$ is odd, gathering is always feasible with three or four colors. The proposed algorithms do not require additional assumptions, such as knowledge of the number of robots, multiplicity detection capabilities, or the assumption of towerless initial configurations. These results demonstrate the power of lights to achieve gathering of robots with limited visibility.

%% file: introduction.tex
\subsection{Background and Motivation}

A lot of research about autonomous mobile robots coordination has been conducted by the distributed computing community. The common goal of these research is to clarify the minimum capabilities for robots to achieve a given task. Hence, most work adopts weak assumptions such as: robots are identical (\emph{i.e.}, robots execute the same algorithm and cannot be distinguished), oblivious (\emph{i.e.}, robots have no memory to record past actions), and silent (\emph{i.e.}, robots cannot send messages to other robots). In addition, to model the behavior of robots, most work uses the Look-Compute-Move (LCM) model introduced by Suzuki and Yamashita~\cite{SY99}. In the LCM model, each robot repeats executing cycles of Look, Compute and Move phases. During the Look phase, the robot takes a snapshot to observe the positions of other robots. According to this snapshot, the robot computes the next movement during the Compute phase. If the robot decides to move, it moves to the target position during the Move phase. By using the LCM model, it is possible to clarify problem solvability both continuous environments (\emph{i.e.}, two- or three-dimensional Euclidean space) and discrete environments (\emph{i.e.}, graph networks). State-of-the-art surveys are given in the recent book by Flocchini \emph{et al.}~\cite{FPS19}.

In this paper, we focus on gathering in graph networks. The goal of gathering is to make all robots gather at a non-predetermined single node. Since gathering is a fundamental task of mobile robot systems and a benchmark application, numerous algorithms have been proposed for various graph network topologies. In particular, many papers focus on ring-shaped networks because symmetry breaking becomes a core difficulty, and any such solution is likely to adapt well on other topologies, as it is possible to make virtual rings over arbitrary networks and hence use ring algorithms in such networks \cite{KKN10,KMP08,IIKO13,DSN14,DNN17,DSKN16}. 

Klasing \emph{et al.}~\cite{KKN10,KMP08} proposed gathering algorithms for rings with \emph{global-weak} multiplicity detection. Global-weak multiplicity detection enables a robot to detect whether the number of robots on each node is one, or more than one. However, the exact number of robots on a given node remains unknown if there is more than one robot on the node.
Izumi \emph{et al.}~\cite{IIKO13} provided a gathering algorithm for rings with \emph{local-weak} multiplicity detection under the assumption that the initial configurations are non-symmetric and non-periodic, and that the number of robots is less than half the number of nodes. Local-weak multiplicity detection enables a robot to detect whether the number of robots on its \emph{current} node is one, or more than one. 
D'Angelo \emph{et al.}~\cite{DSN14,DNN17} proposed unified ring gathering algorithms for most of the solvable initial configurations, using global-weak multiplicity detection \cite{DSN14}, or local-weak multiplicity detection \cite{DNN17}. 
Finally, Klasing \emph{et al.}~\cite{DSKN16} proposed gathering algorithms for grids and trees.
All aforementioned work assumes \emph{unlimited visibility}, that is, each robot can take a snapshot of the whole network graph with all occupied positions.

The unlimited visibility assumption somewhat contradicts the principle of weak mobile robots, hence several recent studies focus on \emph{myopic} robots~\cite{DLLP13,DLLP15,GP13,GP13-1,KLO14}. A myopic robot has limited visibility, that is, it can take a snapshot of nodes (with occupying robots) only within a certain fixed distance $\phi$. Not surprisingly, many problems become impossible to solve in this setting, and several strong assumptions have to be made to enable possibility results.
Datta \emph{et al.}~\cite{DLLP13,DLLP15} study the problem of ring exploration with different values for $\phi$.
Guilbault \emph{et al.}~\cite{GP13} study gathering in bipartite graphs with the global-weak multiplicity detection (limited to distance $\phi$) in case of $\phi=1$, and prove that gathering is feasible only when robots form a star in the initial configuration. They also study the case of infinite lines with $\phi>1$ \cite{GP13-1}, and prove that no universal algorithm exists in this case. In the case of rings, since a ring with even nodes is also a bipartite graph, gathering is feasible only when three robots occupy three successive nodes. For this reason, Kamei \emph{et al.}~\cite{KLO14} give gathering algorithms for rings with $\phi\ge 1$ by using strong assumptions, such as knowledge of the number of robots, and strong multiplicity detection, which enables a robot to obtain the exact number of robots on a particular node.
Overall, limited visibility severely hinders the possibility of gathering oblivious robots on rings.

On the other hand, completely oblivious robots (that can not remember past actions) may be too weak of a hypothesis with respect to a possible implementation on real devices, where persistent memory is widely available. Recently, enabling the possibility that robots maintain a non-volatile visible light~\cite{DFPSY16} has attracted a lot of attention to improve the task solvability. A robot endowed with such a light is called a \emph{luminous} robot. Each luminous robot can emit a light to other robots whose color is chosen among a set of colors whose size is constant. The light color is non-volatile, and so it can be used as a constant-space memory. Viglietta~\cite{G13} gives a complete characterization of the rendezvous problem (that is, the gathering of two robots) on a plane using two visible colored lights assuming unlimited visibility robots. Das \emph{et al.}~\cite{DFPSY16} prove that unlimited visibility robots on two-dimensional space with a five-color light have the same computational power in the asynchronous and semi-synchronous models. Di Luna \emph{et al.}~\cite{DG19} discuss how lights can be used to solve some classical distributed problems such as rendezvous and forming a sequence of patterns. The robots they assume have unlimited visibility, but they also discuss the case where robots visibility is limited by the presence of obstruction. 
Hence, luminous robots seem to dramatically improve the possibility to solve tasks in the LCM model.

As a result of the above observations, it becomes interesting to study the interplay between the myopic and luminous properties for LCM robots: can lights improve task solvability of myopic robots. 
To our knowledge, only three papers \cite{SF-Ex,NOI19,BDL19} consider this combination.
Ooshita \emph{et al.}~\cite{SF-Ex} and Nagahama \emph{et al.}~\cite{NOI19} demonstrate that for the task of ring exploration, even a two-color light significantly improves task solvability. Bramas \emph{et al.}~\cite{BDL19} give exploration algorithms for myopic robots in infinite grids by using a constant-color light. To this day, the characterization of gathering feasibility for myopic luminous robots (aside from the trivial case where a single color is available) is unknown. 

\subsection{Our Contributions}

\begin{figure}
   \centering
    \includegraphics[scale=0.22]{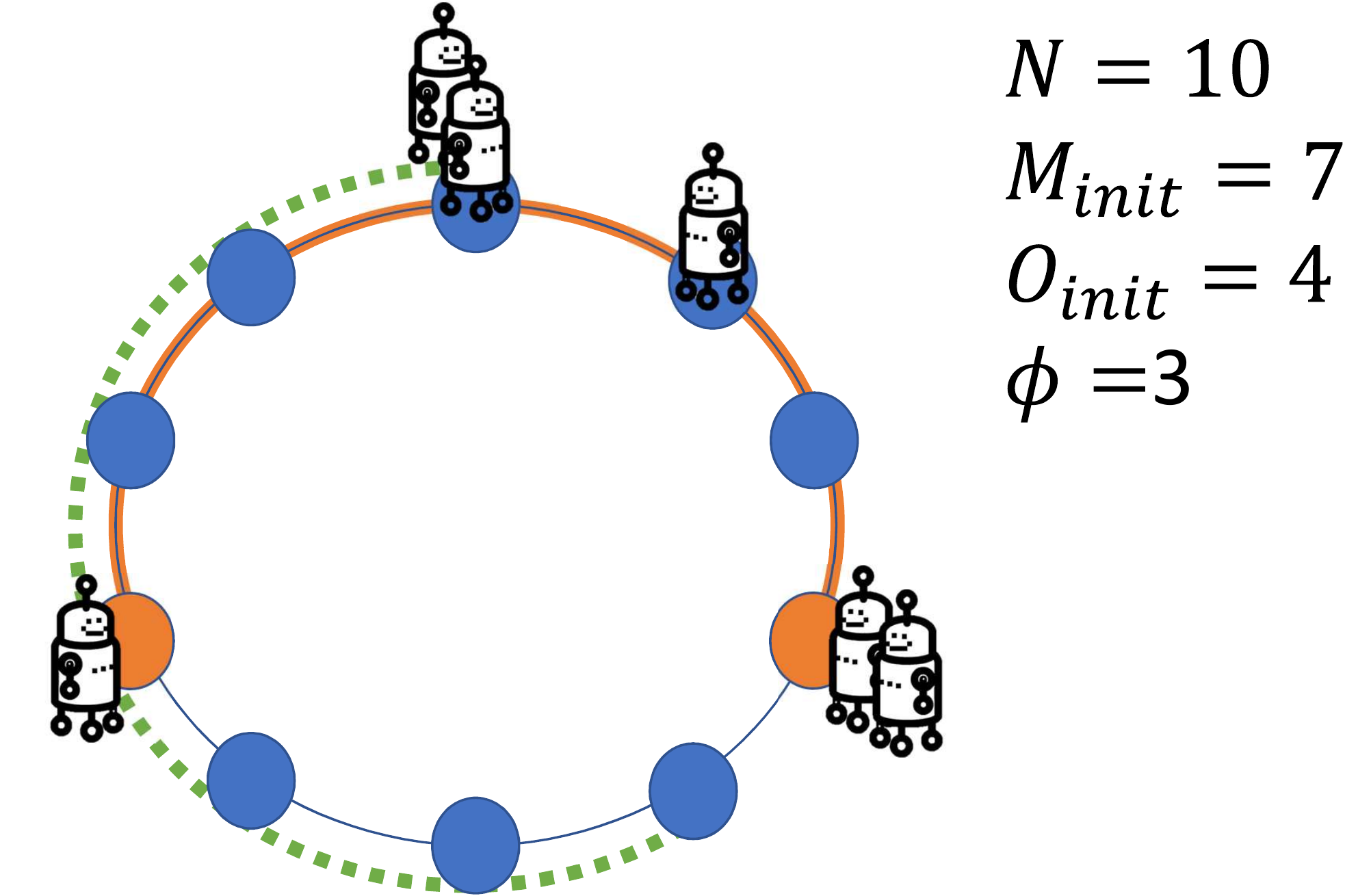}
    \caption{Example of an initial configuration where $M_{init}=7$ and $O_{init}=4$. Border nodes are represented in orange color. The left border robot can sense the range represented by green dots.}
    \label{fig:my_label}
\end{figure}

\begin{table}[t]
\caption{Summary of gathering algorithms for myopic robots in rings. All algorithms assume that the visibility graph is connected, and that there exist two borders. $R$ denotes the number of robots and $H_{center}$ denotes the size of the center hole in the initial configuration. The global (resp., local) strong multiplicity detection enables a robot to observe the exact number of robots on every node (resp., on its current node).}
\label{tab:results}
\centering
\scalebox{0.9}{
\begin{tabular}{|c|c|c|c|c|} 
\hline
 & $\phi$ & \#colors & Multiplicity detection & Assumptions \\
\hline \hline
\cite{KLO14} & $\ge 1$ & 1 & local strong & $R$ is odd \\
             &       &   &        & $R$ is known \\
             &       &   &        & towerless initial configuration\\
\hline
\cite{KLO14} & $\ge 3$ & 1 & global strong & $R$ is even \\
             &       &   &               & $R$ is known \\
             &       &   &               & $H_{center}$ is odd \\
             &       &   &               & $H_{center} \le \phi-2$ \\
             &       &   &               & towerless initial configuration \\
\hline
Algorithm 1  & $\ge 1$ & 3 & none & $M_{init}$ is odd \\
\hline
Algorithm 2  & $\ge 2$ & 4 & none & $O_{init}$ is odd \\
\hline
\end{tabular}}
\end{table}

We clarify the solvability of gathering for myopic luminous robots in rings. 
We consider the asynchronous (ASYNC) model, which is the most general timing assumption. As in previous work by Kamei \emph{et al.}~\cite{KLO14}, we focus on initial configurations such that the visibility graph\footnote{A visibility graph is defined as ${\cal G}_V=({\cal R},{\cal E}_V)$ where ${\cal R}$ is a set of all robots, and ${\cal E}_V$ is a set of pairs of robots that can observe each other.} is connected and there exist two border nodes\footnote{Node $v$ is a border node if robots on $v$ can observe other robots only in one direction.} (see Figure\,\ref{fig:my_label} for an example). Both assumptions are necessary for the class of \emph{cautious} gathering algorithms (see Lemmas~\ref{lem:morethantwoborders} and \ref{lem:noborder}). A cautious gathering protocol never expands the span of a given visibility graph.
In addition, we assume that all robots have the same color in initial configurations. 
We partition initial configurations using two parameters $M_{init}$ and $O_{init}$; 
$M_{init}$ is defined as the number of nodes between two border nodes,
and $O_{init}$ is defined as the number of nodes occupied by some robots (also see Figure\,\ref{fig:my_label}). We can easily observe that, if both $M_{init}$ and $O_{init}$ are even, there exist (so-called \emph{edge-symmetric}) initial configurations such that no algorithm achieves gathering (Corollary~\ref{cor:imp-ee}). Hence, we consider the case that $M_{init}$ or $O_{init}$ is odd.

On the positive side, our main contribution is to prove that, if $M_{init}$ or $O_{init}$ is odd, gathering is always feasible by using a constant number of colors without additional assumptions (so, no multiplicity detection is necessary) for any positive visible distance $\phi$. First, for the case that $M_{init}$ is odd and $\phi\ge 1$ holds, we give a gathering algorithm that uses three colors. Second, for the case that $O_{init}$ is odd and $\phi \ge 2$, we give a gathering algorithm that uses four colors. Note that we assume $\phi \ge 2$ in the second algorithm because, if $\phi=1$ holds, then $O_{init}=M_{init}$ also holds from the assumption of connected visibility graphs, so the first algorithm can be used in this case. We compare the current work with that of Kamei \emph{et al.}~\cite{KLO14} in Table \ref{tab:results}. Overall, lights with a constant number of colors permit to remove most of the previously considered assumptions. For example, our algorithms do not require any multiplicity detection (that is, robots do not distinguish whether the number of robots with the same color on a single node is one or more than one). Furthermore, our algorithms solve gathering even if initial configurations include tower nodes (a tower node is a node that hosts multiple robots). These results demonstrate the power of lights to achieve mobile robot gathering with limited visibility.

%% file: model.tex

We consider anonymous, disoriented and undirected rings $\mathcal{G}$ of $N (\geq 3)$ nodes $u_0, u_1, \dots, u_{N-1}$ such that $u_i$ is connected to both $u_{((i-1)\mod N)}$ and $u_{((i+1)\mod N)}$.
On this ring, $R$ autonomous robots collaborate to gather at one of the $N$ nodes of the ring, not known beforehand, and remain there indefinitely. 

The \emph{distance} between two \emph{nodes} $u$ and $v$ on a ring is the number of edges in a shortest path connecting them.
The \emph{distance} between two \emph{robots} $r_1$ and $r_2$ is the distance between two nodes occupied by $r_1$ and $r_2$, respectively.
Two robots or two nodes are \emph{neighbors} if the distance between them is one.

Robots are identical, \emph{i.e.}, they execute the same program and use no localized parameter such as an identifier or a particular orientation. Also, they are oblivious, \emph{i.e.}, they cannot remember their past observations or actions. 
We assume that robots do \emph{not} know $N$, the size of the ring, and $R$, the number of robots.

Each robot $r_i$ maintains a variable $L_i$, called {\em light}, which spans a finite set of states called \emph{colors}. A light is \emph{persistent} from one computational cycle to the next: the color is not automatically reset at the end of the cycle. Let $K$ denote the number of available light colors.
Let $L_i(t)$ be the light color of $r_i$ at time $t$. 
We assume the \emph{full light} model: each robot $r_i$ can see the light of other robots, but also its own light.
Robots are unable to communicate with each other explicitly (\emph{e.g.}, by sending messages), however, they can observe their environment, including the positions and colors of the other robots. 
We assume that besides colors, robots do \emph{not} have multiplicity detection capability: if there are multiple robots $r_1, r_2,\ldots r_k$ in a node $u$, an observing robot $r$ can detect only colors, so $r$ can detect there are multiple robots at $u$ if and only if at least two robots among $r_1, r_2,\ldots r_k$ have different colors. So, a robot $r$ observing a single color at node $u$ cannot know how many robots are located in $u$.

Based on the sensing result, a robot $r$ may decide to move or to stay idle. 
At each time instant $t_i(1 \leq i)$, robots occupy nodes of the ring, their positions and colors form a \emph{configuration} $C(t_i)$ of the system at time $t_i$. 
When $C(t_i)$ reaches $C(t_{i+1})$ by executing some phases between $t_i$ and $t_{i+1}$, it is denoted as $C(t_i) \longrightarrow C(t_{i+1})$.
The reflexive and transitive closure is denoted as $\longrightarrow^*$.

We assume that robots have limited visibility: an observing robot $r$ at node $u$ can only sense the robots that occupy nodes within a certain distance, denoted by $\phi$ ($\phi \geq 0$), from $u$. As robots are identical, they share the same $\phi$.

Let $\mathcal{X}_i(t)$ be the set of colors of robots located in node $u_i$ at time $t$. 
If a robot $r_j$ located at $u_i$ at $t$, the sensor of $r_j$ outputs a sequence, $\mathcal{V}_j$, of $2\phi+1$ set of colors: 
$$\mathcal{X}_{i-\phi}(t), \ldots , \mathcal{X}_{i-1}(t), (\mathcal{X}_i(t)), \mathcal{X}_{i+1}(t), \ldots, \mathcal{X}_{i+\phi}(t).$$
This sequence $\mathcal{V}_j$ is the \emph{view} of $r_j$. If the sequence $\mathcal{X}_{i+1}, \ldots , \mathcal{X}_{i+\phi}$ is equal to the sequence  $\mathcal{X}_{i-1}, \ldots , \mathcal{X}_{i-\phi}$, then the view $\mathcal{V}_j$ of $r_j$ is \emph{symmetric}.
Otherwise, it is \emph{asymmetric}.  
In $\mathcal{V}_j$, a node $u_k$ is \emph{occupied} at instant $t$ whenever $|\mathcal{X}_k(t)|>0$.
Conversely, if $u_k$ is not occupied by any robot at $t$, then $\mathcal{X}_{k}(t)=\emptyset$ holds, and $u_k$ is \emph{empty} at $t$.

If there exists a node $u_i$ such that $|\mathcal{X}_{i}(t)|=1$ holds, $u_i$ is \emph{singly-colored}.
Note that $|\mathcal{X}_{i}(t)|$ denotes the number of colors at node $u_i$, thus even if $u_i$ is singly-colored, it may be occupied by multiple robots (sharing the same color).
Now, if a node $u_i$ is such that $|\mathcal{X}_{i}(t)|>1$ holds, $u_i$ is \emph{multiply-colored}. As each robot has a single color, a multiply-colored node always hosts more than one robot.

In the case of a robot $r_j$ located at a singly-colored node $u_i$, $r_j$'s view $\mathcal{V}_j$ contains an $\mathcal{X}_i(t)$ that can be written as $[L_j]$.
Then, if the left node of $u_i$ contains one or more robots with color $L_k$, and the right node of $u_i$ contains one or more robots with color $L_l$, while $u_i$ only hosts $r_j$, then $\mathcal{V}_j$ can be written as $L_k[L_j]L_l$. 
Now, if robot $r_j$ at node $u_i$ occupies a multiply-colored position (with two other robots $r_k$ and $r_l$ having distinct colors), then $|\mathcal{X}_i(t)|=3$, and we can write $\mathcal{X}_i(t)$ in $\mathcal{V}_j$ as 
$\begin{bmatrix}
L_k\\
L_l\\
[L_j]
\end{bmatrix}$.
When the view does not consist of a single observed node, we use brackets to distinguish the current position of the observing robot in the view and the inner bracket to explicitly state the observing robot's color.

Our algorithms are driven by observations made on the current view of a robot, so many instances of the algorithms we use \emph{view predicates}: a Boolean function based on the current view of the robot.
The predicate $L_j$ matches any set of colors that includes color $L_j$, while predicate $(L_j,L_k)$ matches any set of colors that contains $L_j$, $L_k$, or both. Now the predicate  
$\begin{pmatrix}
L_1\\
L_2
\end{pmatrix}$ matches any set that contains \emph{both} $L_1$ and $L_2$. Some of our algorithm rules expect that a node is singly-colored, \emph{e.g.} with color $L_k$, in that case, the corresponding predicate is denoted by $L_k!$.
To express predicates in a less explicit way, we use character '?' to represent any non-empty set of colors, so a set of colors $\mathcal{X}_i\neq\emptyset$ satisfies predicate '?'. The $\neg$ operator is used to negate a particular predicate $P$ (so, $\neg P$ returns \emph{false} whenever $P$ returns \emph{true} and vice versa).
Also, the superscript notation $P^{y}$ represents a sequence of $y$ consecutive sets of colors, each satisfying predicate $P$. Observe that $y \leq \phi$.
In a given configuration, if the view of a robot $r_j$ at node $u_i$ satisfies predicate ${\emptyset}^{\phi}[?]$ or predicate $[?]{\emptyset}^{\phi}$, then $r_j$ is a \emph{border robot} and $u_i$ a \emph{border node}.
Sometimes, we require a particular color $L_1$ to be present at some position $u_i$ and a particular color $L_2$ \emph{not} to be present at some position $u_i$. 
For the above case, the corresponding predicate would be:
$\begin{pmatrix}
L_1\\
\neg L_2 
\end{pmatrix}$.

In this paper, we aim at maintaining the property that at most two border nodes exist at any time.
On the ring $\mathcal{G}$, let $H_{max}$ be the size of the maximum hole (\emph{i.e.}, the maximum sequence of empty nodes).
Note that by the assumptions, at instant $t=0$ (\emph{i.e.}, in the initial configuration), $H_{max}>\phi$ holds. 
Let $\mathcal{V}^\prime$ be the subset of nodes on a path between two border nodes $u$ and $v$, such that all robots are hosted by nodes in $\mathcal{V}^\prime$. Also, let $\mathcal{G}^\prime$ be the subgraph of $\mathcal{G}$ induced by $\mathcal{V}^\prime$. 
Note that, $\mathcal{G}^\prime$ does not include the hole with the size $H_{max}$.
At instant $t=0$, let $H_{init}$ be the maximum distance between occupied nodes in $\mathcal{G}^\prime$, $M_{init}$ be the number of nodes in $\mathcal{G}^\prime$, and $O_{init}(\leq M_{init})$ be the number of occupied nodes in $\mathcal{G}^\prime$.
We assume that $\phi\geq H_{init}\geq 1$ holds.
Note that, $H_{init}$ is the size of the second maximum hole in the ring because there are two border nodes. 
As previously stated, no robot is aware of $H_{init}$, $M_{init}$ and $O_{init}$.
In $\mathcal{G}^\prime$, let $D$ denote the distance between the two border nodes.
Note that, at $t=0$, $D=M_{init}-1$ holds.

Each robot $r$ executes Look-Compute-Move cycles infinitely many times: $(i)$ first, $r$ takes a snapshot of the environment and obtains an ego-centered view of the current configuration (Look phase), $(ii)$ according to its view, $r$ decides to move or to stay idle and possibly changes its light color (Compute phase), $(iii)$ if $r$ decided to move, it moves to one of its neighbor nodes depending on the choice made in the Compute phase (Move phase).
At each time instant $t$, a subset of robots is activated by an entity known as the scheduler. This scheduler is assumed to be fair, \emph{i.e.}, all robots are activated infinitely many times in any infinite execution. In this paper, we consider the most general asynchronous model: the time between Look, Compute, and Move phases is finite but unbounded. We assume however that the move operation is atomic, that is, when a robot takes a snapshot, it sees robots colors on nodes and not on edges. 
Since the scheduler is allowed to interleave the different phases between robots, some robots may decide to move according to a view that is different from the current configuration. Indeed,  during the compute phase, other robots may move. Both the view and the robot are in this case said to be \emph{outdated}.


In this paper, each rule in the proposed algorithms is presented in the similar notation as in \cite{SF-Ex}: 
$<Label>$ $:$ $<Guard>$ $::$ $<Statement>$. 
The guard is a predicate on the view $\mathcal{V}_j = \mathcal{X}_{i-\phi}, \ldots ,\mathcal{X}_{i-1}, (\mathcal{X}_i),
\mathcal{X}_{i+1}, \ldots , \mathcal{X}_{i+\phi}$
obtained by robot $r_j$ at node $u_i$ during the Look phase. If the predicate evaluates to \emph{true}, $r_j$ is \emph{enabled}, otherwise, $r_j$ is \emph{disabled}.
In the first case, the corresponding rule $<Label>$ is also said to be {\em enabled}.
If a robot $r_j$ is enabled, $r_j$ may change its color and then move based on the corresponding statement during its subsequent Compute and Move phases.
The statement is a pair of ({\it New color}, {\it Movement}).
Movement can be
($i$) $\rightarrow$, meaning that $r_j$ moves towards node $u_{i+1}$, 
($ii$) $\leftarrow$, meaning that $r_j$ moves towards node $u_{i-1}$, and
($iii$) $\bot$, meaning that $r_j$ does not move.
For simplicity, when $r_j$ does not move (resp. $r_j$ does not change its color), we omit {\it Movement} (resp. {\it New color}) in the statement.
The label $<Label>$ is 
denoted as R followed by a non-negative
integer (\emph{i.e.}, R0, R1, etc.) where a smaller label indicates higher priority.

%% file: algorithm1-full.tex
The strategy of our algorithm is as follows:
The robots on two border nodes keep their lights Red or Blue, then the algorithm can recognize that they are originally on border nodes.
When robots on a border node move toward the center node, they change the color of their light to Blue or Red alternately regardless of the neighboring nodes being occupied, where initially robots become Red colors.
To keep the connected visibility graph, when a border node becomes singly-colored, the border robot changes its light to Blue or Red according to the distance from the original border node and moves toward the center node and the neighboring non-border robot becomes a border robot.
Eventually, two border nodes become neighboring.
Then, one has Blue robots and the other has Red robots 
because $M_{init}$ is odd. 
At the last moment, Red robots join Blue robots to achieve the gathering.

The formal description of the algorithm is in Algorithm~\ref{alg1}.
The rules of our algorithm are as follows:
\begin{itemize}
\item R0: If the gathering is achieved, a robot does nothing\footnote{
{Note, this algorithm and the next one cannot terminate because gathering configurations are not terminating ones due to robots with outdated views even if this rule is executed. Because this rule has higher priority, if it is enabled, robots do not need to check other guards.}}.
\item R1: A border White robot on a singly-colored border node changes its light to Red.
\item R2a and R2b: A border Red robot on a singly-colored border node changes its light to Blue and moves toward an occupied node.
\item R3a and R3b: A border Blue robot on a singly-colored border node changes its light to Red and moves toward an occupied node.
\item R4a and R4b: When White robots become border robots, they change their color to the same color as the border Red or Blue robots.
\item R5: If two border nodes are neighboring, a border Red robot on a singly-colored border node moves to the neighboring singly-colored node with Blue robots.
\end{itemize}

Figure \ref{fig:ex-alg1} illustrates an execution example of Algorithm \ref{alg1}. This figure assumes $\phi=2$. Figure \ref{fig:ex-alg1}(a) shows an initial configuration. First, border White robots change their lights to Red by R1 (Fig.\,\ref{fig:ex-alg1}(b)). Next, left border Red robots move by R2a. Since we consider the ASYNC model, some robots may become outdated. In Fig.\,\ref{fig:ex-alg1}(c), the top robot has changed its light but not yet moved, and the middle robot has looked but not yet changed its light. The outdated robots move in Fig.\,\ref{fig:ex-alg1}(d), and then the right border Red robot also moves in Fig.\,\ref{fig:ex-alg1}(e) by R2a. Here, one left border Red robot has not yet moved. However, since it still observes a White robot, it can move by R2a (Fig.\,\ref{fig:ex-alg1}(f)). Then, border Blue robots can move by R3b. In Fig.\,\ref{fig:ex-alg1}(g), one left border robot becomes outdated and the right border robot completes the movement. After the right border White robot changes its light to Red by R4a (Fig.\,\ref{fig:ex-alg1}(h)), the right border Red robots move by R5 (Fig.\,\ref{fig:ex-alg1}(i)). Note that, all robots stay at a single node but one of them is outdated. Hence the robot moves after that (Fig.\,\ref{fig:ex-alg1}(j)), but it can go back to the gathering node by R5 (Fig.\,\ref{fig:ex-alg1}(k)). Now robots have achieved gathering.

\begin{figure}[t]
\centering\includegraphics[keepaspectratio, width=11cm]{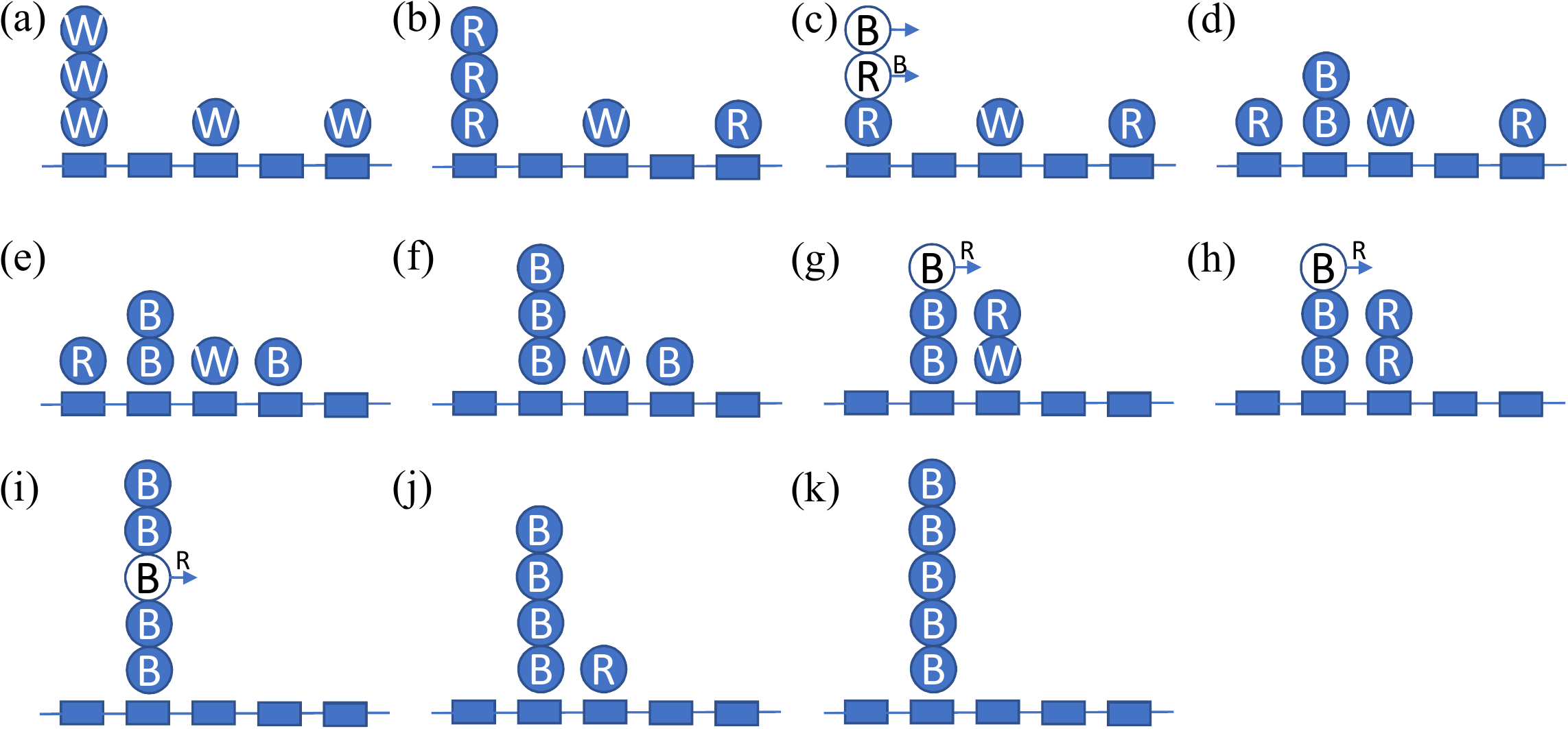}
\caption{An execution example of Algorithm \ref{alg1} with $\phi=2$.}
\label{fig:ex-alg1}
\end{figure}

In the case of $\phi=1$, each robot can view only its neighboring nodes.
In this case, we should add some assumptions as follows:
\begin{itemize}
    \item R2a and R3a are always disabled.
    \item R2b and R3b are enabled when there are White robots in the neighboring node.
    \item R5 is enabled when the neighboring node is singly-colored with Blue robots.
\end{itemize}
In that case, because of the connectivity of the visibility graph, there is no empty node in $\mathcal{G}^\prime$. That is, until gathering is achieved, there is at least one White robot on the neighboring node. Thus, such assumption is natural and we can prove the correctness in the same way as other cases by deleting R2a and R3a (the case such that there is no White robot on the neighboring node but gathering is not achieved).

\begin{algorithm}[t] 
\input{alg1-full.tex}

\caption{Algorithm for the case that $M_{init}$ is odd.}
\label{alg1}
\end{algorithm}

%% file: alg1-full.tex
{\bf Colors}\\
W (White), R (Red), B (Blue)\\
\\
{\bf Rules}\\
/* Do nothing after gathering. */\\
R0: $\emptyset^\phi[?]\emptyset^\phi$ :: $\bot$\\

/* Start by the initial border robot. */\\
R1: $\emptyset^\phi[W!](\neg\emptyset^{\phi})$ :: $R$\\

/* Border robots on singly-colored nodes change their color alternately and move. */\\
R2a: $\emptyset^\phi[R!](\emptyset,B!)(\neg(\emptyset^{\phi-1}))$ :: $B, \rightarrow$\\

R2b: $\emptyset^\phi[R!](W)(?^{\phi-1})$ :: $B, \rightarrow$\\

R3a: $\emptyset^\phi[B!](\emptyset,R!)(\neg(\emptyset^{\phi-1}))$ :: $R, \rightarrow$\\

R3b: $\emptyset^\phi[B!](W)(?^{\phi-1})$ :: $R, \rightarrow$\\

/* When White robots become a border robot, they change their color to the same color as the border robot. */\\
R4a: $\emptyset^\phi\begin{bmatrix}
R\\
[W]
\end{bmatrix}(\neg\emptyset^{\phi})$ :: $R$\\

R4b: $\emptyset^\phi\begin{bmatrix}
B\\
[W]
\end{bmatrix}(\neg\emptyset^{\phi})$ :: $B$\\

/* When two border nodes are neighboring, robots gather to a node with the Blue border robots. */\\
R5: $\emptyset^\phi[R!](B!)(\emptyset^{\phi-1})$ :: $B, \rightarrow$

%% file: alg1-correctness-full.tex
\begin{lemma}[Lemma 1 in the main part]\label{NB}
When a robot $r_i$ looks, if it is a non-border, it cannot execute any action.
\end{lemma}
\begin{proof}
There is no rule such that $r_i$ can execute by the definition of Algorithm 1.
Hence, the lemma holds.
\end{proof}


To discuss the correctness, we consider the time instants such that the distance $D$ between the borders has just reduced at least one. The duration between them is called {\em mega-cycle}. Note that, $D$ is reduced at most two by the algorithm during a mega-cycle.
Let $t_0, t_1, \ldots t_i, \dots$ be starting times of mega-cycles, where $t_0$ is the starting time of the algorithm and for each $i~(\geq 1)$, $D$ is reduced at least one from $t_{i-1}$ at time $t_i$. Letting $C(t_i)$ be the configuration at time $t_i$, the transition of configurations from $t_i$ to $t_{i+1}$ is denoted as $C(t_i) \xrightarrow{MC} C(t_{i+1})$. 



Figure~\ref{TD} shows a transition diagram of configurations for each mega-cycle (We prove all transitions later).
The small blue box represents a node, and the circle represents a set of robots (One single circle may represent a set of robots with the same color).
The doubly (resp. singly) lined arrows represent that $D$ decreases by 2 (resp. 1).  
The letter in each circle represents the color of the lights.
The circles in parentheses represent that they are optional.
The circles in brackets represent that one of them should exist.
In {\sf Conf 1-3}, 
there is no Red or Blue robot with an outdated view (i.e., all robots are after Move phases before they look).
In {\sf Conf 4-7}, left side border represents 
that there is no Red or Blue robot with an outdated view.
Right side borders in {\sf Conf 4-7} represent that they are still working on their movements, i.e., there may be robots with outdated views.
The second node from the right can be empty, and there can be White robots on the node.
In the right side borders, the white circles represent that the robots can be on the node, but with the outdated view.
Actually, there may be White robots which do not change their color to the same as the border color, Red or Blue, yet. 
We omit such White robots changing to border color for simplicity, i.e. they may be included in any set of White robots in border nodes in this figure\footnote{However, it is considered in the proof.}.
For example, consider an example in Fig.\,\ref{fig:ex-alg1}. The initial configuration in Fig.\,\ref{fig:ex-alg1}(a) is represented by {\sf Init}. The next mega-cycle starts in Fig.\,\ref{fig:ex-alg1}(e), and this configuration is represented by {\sf Conf 4}. Note that mapping from left and right borders in {\sf Conf 1-7} to two borders in configurations may change during an execution. The next mega-cycle starts in Fig.\,\ref{fig:ex-alg1}(f), and this configuration is represented by {\sf Conf 1}.

\begin{figure*}[t]
\centering\includegraphics[keepaspectratio, width=13cm]{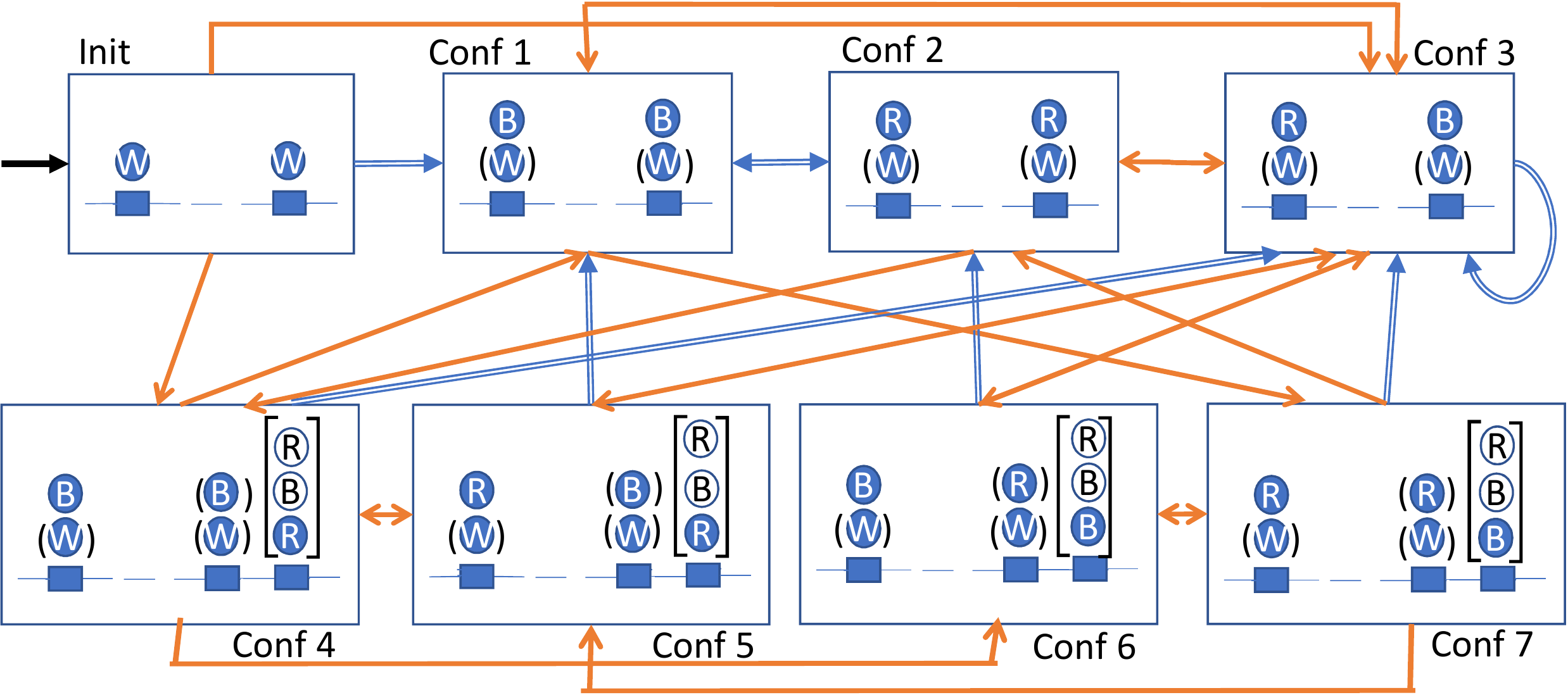}
\caption{Transition Diagram for Algorithm 1 while $D>1$.}\label{TD}
\end{figure*}

\begin{lemma}[Lemma 2 in the main part]\label{mul2sing}
Assume that a border node contains White and $\gamma$ robots, where $\gamma \in \{$Red, Blue$\}$ at time $t$. Then there exists a time $t' (>t)$ such that the border node becomes singly-colored one with $\gamma$ robots. 
In the case that the border node contains only White robots, the border node becomes Red singly-colored.
\end{lemma}
\begin{proof}
First, consider the case that the border node $u$ contains only White robots.
By the definition of Algorithm 1, border robots on $u$ can execute only R1.
Then, in $u$, some of them may change their light to Red, some of them may remain White with outdated view, and others do not look yet.
White robots with outdated view eventually change their light to Red.
White robots which look Red robots on the node can execute R4a.
Thus, every robot on $u$ becomes Red, that is, the border node becomes Red singly-colored.

Next, consider the case that a border node $u$ contains White and $\gamma$ robots at time $t$.
If $\gamma=Red$, White robots on $u$ can execute only R4a and change their color to Red.
If $\gamma=Blue$, White robots on $u$ can execute only R4b and change their color to Blue.
Then, in both cases, $\gamma$ robots on $u$ cannot execute any Rule since $u$ is not singly-colored.
Thus, every robot on $u$ becomes $\gamma$, that is, the border node becomes singly-colored.
\end{proof}

\begin{lemma}[Lemma 3 in the main part]\label{Reduction-by-one}
Assume that a border node $u$ is singly-colored with Red (resp. Blue) at time $t$, then if there is no robot in $u$ at time $t' (>t)$, all robots in $u$ change their color to Blue (resp. Red), move to the neighboring node of $u$, and the distance between the border nodes is reduced by at least one at $t'$. 
\end{lemma}
\begin{proof}
At time $t$, every robot on $u$ can execute only R2 (resp. R3) because $u$ is singly-colored.
On $u$, by R2 (resp. R3), some of them may change their light to Blue (resp. Red), some of them may remain Red (resp. Blue) with outdated views, and others do not look yet.
Blue (resp. Red) robots eventually move to the neighboring node $v$.
Red (resp. Blue) robots with outdated views eventually change their lights to Blue (resp. Red) and move to $v$.
Red (resp. Blue) robots $A$ which look Blue (resp. Red) robots on $u$ cannot execute any rule by the definition of Algorithm 1.
However, because Blue (resp. Red) robots on $u$ move to $v$ eventually, $u$ becomes singly-colored.
If $v$ is occupied in the initial configuration, $A$ can observe a White robot on $v$. 
Otherwise, $A$ can observe other robots 
beyond $v$ because robots initially located on $u$ have moved to $v$. 
Therefore, $A$ also can execute R2 (resp. R3), the border position moves to $v$ and the border is with Blue (resp. Red) robots. 
Thus, the lemma holds.
\end{proof}

\begin{lemma}\label{Init}
{\sf Init} $\xrightarrow{MC}$ {\sf Conf 1},  {\sf Init} $\xrightarrow{MC}$ {\sf Conf 3} or {\sf Init} $\xrightarrow{MC}$ {\sf Conf 4}.
\end{lemma}
\begin{proof}
In the initial configuration, all robots have White lights.
By Lemma~\ref{NB}, non-border robots cannot execute any action.
Since the border nodes are White singly-colored, there exists a time $t$ such that either border node, say $u$, becomes 
Red singly-colored by Lemma~\ref{mul2sing}.
Since Lemma~\ref{Reduction-by-one} holds at $t$, the border position moves to the neighbor node $v$ and the new border node $v$ contains the all robots in $u$ and their colors are Blue. 
If $v$ is originally occupied in the initial configuration, the border node has Blue robots and White robots just after the first mega-cycle.

From the initial configuration, in the first mega-cycle, 
if both of two borders change their positions, the configuration becomes {\sf Conf 1}.
If one border changes its position but the other is still working on the movement, the configuration becomes {\sf Conf 3} or {\sf Conf 4} because in the other border node
some robots change their color to Red but there remain White robots ({\sf Conf 3}), or on the way to moving the border ({\sf Conf 4}).

Thus, the lemma holds.
\end{proof}

\begin{lemma}\label{Conf1}
{\sf Conf 1} $\xrightarrow{MC}$ {\sf Conf 2},  {\sf Conf 1} $\xrightarrow{MC}$ {\sf Conf 3} or {\sf Conf 1} $\xrightarrow{MC}$ {\sf Conf 7}.
\end{lemma}
\begin{proof}
In the configuration {\sf Conf 1}, both of two borders have Blue and White robots.

First, we consider the case that only border robots on the node $u_i$ execute until the time $t$, the end of the next mega-cycle, and the border robots on $u_j$ remain to have Blue and White robots. 
Then, after $u_i$ becomes Blue singly-colored by Lemma~\ref{mul2sing}, the border position moves to the neighbor node $v$ and the new border contains all the robots located in $u_i$ which becomes Red. 
If the new border node $v$ is originally occupied in the initial configuration, there are Red and White robots in $v$ just after $t$.
Thus, from {\sf Conf 1}, the configuration becomes {\sf Conf 3} if only robots in $u_i$ execute in this mega-cycle.

Next, we consider the case that both of two borders execute in $t$.
In both borders, the same movements we showed above occur. 
In $t$, if both borders complete their movements, the configuration becomes {\sf Conf 2}.
If one border complete their movements and the other is still working on their movements, then the configuration becomes {\sf Conf 7}.

Thus, the lemma holds.
\end{proof}

By the similar proof, we can derive the following lemmas.
\begin{lemma}\label{Conf2}
{\sf Conf 2} $\xrightarrow{MC}$ {\sf Conf 1},  {\sf Conf 2} $\xrightarrow{MC}$ {\sf Conf 3} or {\sf Conf 2} $\xrightarrow{MC}$ {\sf Conf 4}.
\end{lemma}

\begin{lemma}\label{Conf3}
{\sf Conf 3} $\xrightarrow{MC}$ {\sf Conf 1},  {\sf Conf 3} $\xrightarrow{MC}$ {\sf Conf 2}, 
 {\sf Conf 3} $\xrightarrow{MC}$ {\sf Conf 3},
{\sf Conf 3} $\xrightarrow{MC}$ {\sf Conf 5}, or {\sf Conf 3} $\xrightarrow{MC}$ {\sf Conf 6}.
\end{lemma}

\begin{lemma}\label{Conf4}
{\sf Conf 4} $\xrightarrow{MC}$ {\sf Conf 1},  {\sf Conf 4} $\xrightarrow{MC}$ {\sf Conf 3}, 
 {\sf Conf 4} $\xrightarrow{MC}$ {\sf Conf 5},
or {\sf Conf 4} $\xrightarrow{MC}$ {\sf Conf 6}.
\end{lemma}
\begin{proof}
In the configuration {\sf Conf 4}, border robots in $u_i$ have Blue and White robots without outdated views, and the other border robots in $u_j$ are still working on their movements from Red border, that is, there may be robots with outdated views.

In the case that only robots in $u_i$ move in the next mega-cycle $t$ and border robots in $u_j$ remain, by the proof of Lemma~\ref{Conf1}, the configuration becomes {\sf Conf 5}.
Consider the case that only robots in $u_j$ move to the neighboring node $v$ in $t$.
If border in $u_i$ are working on their movement while the border in $v$ becomes Blue border, the configuration becomes {\sf Conf 6}.
If border in $u_i$ does not look yet, then the configuration becomes {\sf Conf 1}.
In the case that both of two borders complete their movements at the same time, then the configuration becomes {\sf Conf 3}.

Thus, the lemma holds.
\end{proof}

By the similar proof, we can derive the following lemma.
\begin{lemma}\label{Conf7}
{\sf Conf 7} $\xrightarrow{MC}$ {\sf Conf 2},  {\sf Conf 7} $\xrightarrow{MC}$ {\sf Conf 3}, 
 {\sf Conf 7} $\xrightarrow{MC}$ {\sf Conf 5},
or {\sf Conf 7} $\xrightarrow{MC}$ {\sf Conf 6}.
\end{lemma}

\begin{lemma}\label{Conf5}
{\sf Conf 5} $\xrightarrow{MC}$ {\sf Conf 3},  {\sf Conf 5} $\xrightarrow{MC}$ {\sf Conf 4}, 
or {\sf Conf 5} $\xrightarrow{MC}$ {\sf Conf 1}.
\end{lemma}
\begin{proof}
In the configuration {\sf Conf 5}, border robots in $u_i$ have Red and White robots without outdated view, and the other border robots in $u_j$ are in the progress of their movements from the Red border.

In the case that only robots in $u_i$ move in the next mega-cycle $t$ and border robots in $u_j$ remain, the configuration becomes {\sf Conf 4}.
Consider the case that only roots in $u_j$ move to the neighboring node $v$ in $t$.
If border in $u_i$ does not look yet, then the configuration becomes {\sf Conf 3}.
If border in $u_i$ become in the process of their movement, the configuration becomes {\sf Conf 4}.
In the case that both of two borders complete their movements at the same time, then the configuration becomes {\sf Conf 1}.

Thus, the lemma holds.
\end{proof}

By the similar proof, we can derive the following lemma.
\begin{lemma}\label{Conf6}
{\sf Conf 6} $\xrightarrow{MC}$ {\sf Conf 2},  {\sf Conf 6} $\xrightarrow{MC}$ {\sf Conf 3}, 
or {\sf Conf 6} $\xrightarrow{MC}$ {\sf Conf 7}.
\end{lemma}

\begin{lemma}[Lemma 4 in the main part]\label{RB3}
From the initial configuration, $D$ decreases monotonically and eventually becomes 2.
\end{lemma}
\begin{proof}
By Lemmas~\ref{Init}-\ref{Conf6}, in each mega-cycle, $D$ decreases by at least one.
Thus, the lemma holds.
\end{proof}

\begin{lemma}[Lemma 5 in the main part]\label{EO}
Let $h$ be the distance from the original border node to a node $u_h$ in $\mathcal{G}^\prime$.
If $h$ is odd (resp. even), a Blue (resp. Red) border robot comes into $u_h$.
\end{lemma}
\begin{proof}
By the proof of Lemma~\ref{Init}, if $h$ is 1, Blue robots come into $u_h$.
By Lemmas~\ref{Conf1}-\ref{Conf6}, they change their color Red and Blue alternately every hop.
Thus, if $h$ is even (resp. odd), the border which comes into $u_h$ has Red (resp. Blue).
Thus, the lemma holds.
\end{proof}

Let {\sf Conf BW-MR} be the configuration with $D=1$ such that there are Blue robots and White robots without outdated views in a border node and there are Red robots and Blue robots with outdated views and Red robots without outdated views in the other border node, where Blue robots with outdated views will move to the other border node and Red robots with outdated views will change their color to Blue and move to the other border node (Fig.~4(a)).

Let {\sf Conf RW-MB} be the configuration with $D=1$ such that there are Red robots and White robots without outdated views in a border node and there are Red robots and Blue robots with outdated views and Blue robots without outdated views in the other border node, where Red robots with outdated views will move to the other border node and Blue robots with outdated views will change their color to Red and move to the other border node (Fig.~4(b)).

\begin{figure}[t]\centering
\begin{minipage}[t]{0.45\linewidth}
    \centering
    \includegraphics[keepaspectratio, scale=0.15]{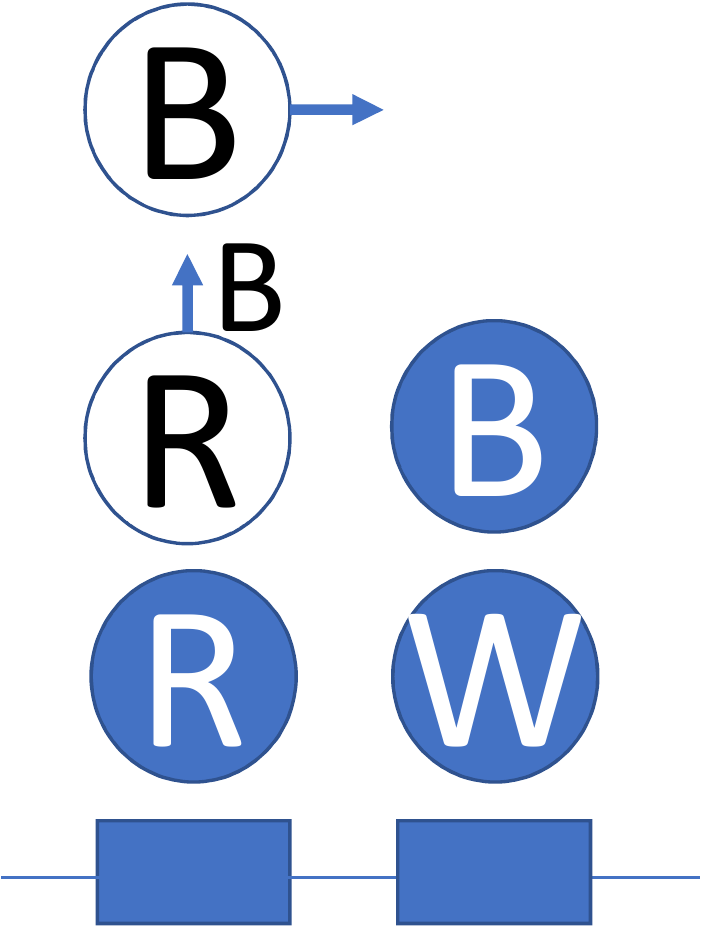}
    \subcaption{{\sf Conf BW-MR}}\label{fig:BWMR}
\end{minipage}
\begin{minipage}[t]{0.45\linewidth}
    \centering
    \includegraphics[keepaspectratio, scale=0.15]{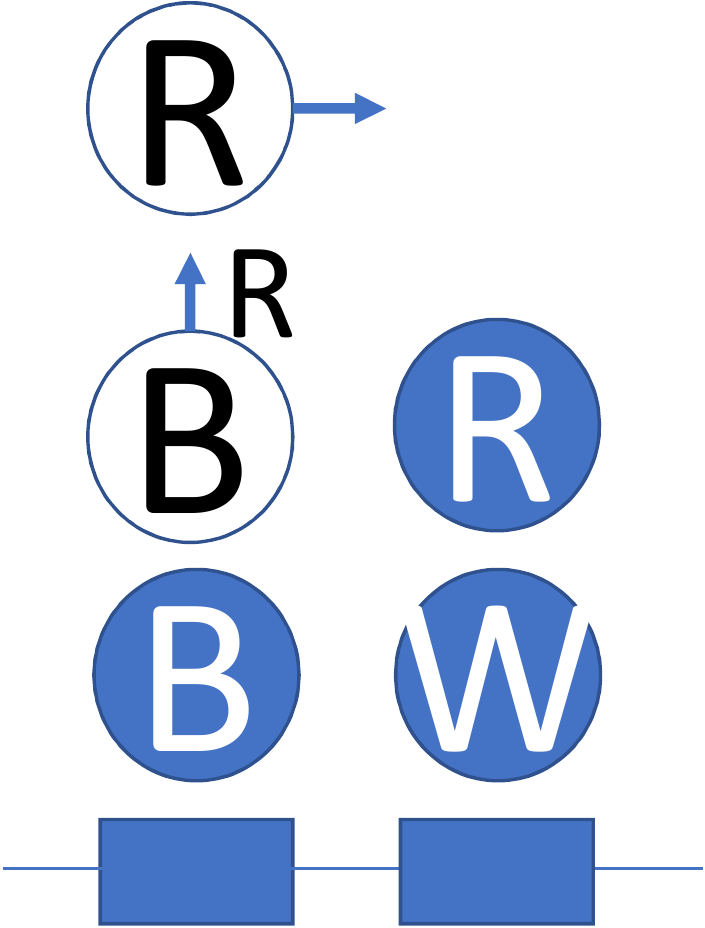}
    \subcaption{{\sf Conf RW-MB}}\label{fig:RWMB}
\end{minipage}
\caption{The configurations {\sf Conf BW-MR} and {\sf Conf RW-MB} with $D=1$} 
\end{figure}


\begin{lemma}[Lemma 6 in the main part]\label{C89}
After $D$ becomes 2, the configuration becomes {\sf Conf BW-MR} or {\sf Conf RW-MB}.
\end{lemma}
\begin{proof}
By Lemma~\ref{RB3}, $D$ eventually becomes two.
In such configuration,
let $A$ and $B$ be two sets of border robots in $u_i$ and $u_k$ respectively, where $u_j$ is neighboring to $u_i$ and $u_k$.
Let $h(u_i)$ (resp. $h(u_k)$) be the distance from the original border node occupied by a part of $A$ (resp. $B$) to $u_i$ (resp. $u_k$) in $\mathcal{G}^\prime$.
Because $M_{init}$ is odd in the initial configuration, if $h(u_i)$ is even (resp. odd), $h(u_k)$ is also even (resp. odd).
By Lemma~\ref{EO}, when $D=2$, the configuration is {\sf Conf 1}, {\sf Conf 2}, {\sf Conf 5} or {\sf Conf 6} in Figure~\ref{TD}.
Without loss of generality, we call left (resp. right) side border node in this figure $u_i$ (resp. $u_k$).

In {\sf Conf 1} (resp. {\sf Conf 2}) such that $D=2$ holds, White robots can execute R4b (resp. R4a).
After at least one of borders becomes a singly-colored node with Blue (resp. Red) robots, then robots in the node can execute R3 (resp. R2) and the configuration becomes {\sf Conf 6} (resp. {\sf Conf 5}) such that $D=2$ holds.

In {\sf Conf 5} such that $D=2$ holds, robots in $u_k$ are executing R2.
Robots in $u_i$ can execute R2 too after every White robot changes their color Red by R4a.
If every robot in $u_k$ finishes executing R2 before robots in $u_i$ start executing R2 (they do no look yet), then the configuration becomes {\sf Conf 3} where $u_i$ and $u_j$ are borders and $D=1$.
Otherwise, if every robot in $u_k$ (resp. $u_i$) finishes executing R2 earlier, then the configuration becomes {\sf Conf BW-MR}.

In {\sf Conf 6} such that $D=2$ holds, robots in $u_k$ are executing R3.
Robots in $u_i$ can execute R3 too after every White robot changes their color Blue by R4b.
If every robot in $u_k$ finishes executing R3 before robots in $u_i$ start executing R3, then the configuration becomes {\sf Conf 3} where $u_i$ and $u_j$ are borders and $D=1$.
Otherwise, if every robot in $u_k$ (resp. $u_i$) finishes executing R3 earlier, then the configuration becomes {\sf Conf RW-MB}.

From {\sf Conf 3} such that $D=1$, only White robots execute R4 until at least one border becomes a singly-colored node.
When both of two border nodes become singly colored, border Blue robots cannot execute any rule by the definition of the algorithm, and border Red robots can execute only R5. Therefore, the gathering is achieved.
When a border node becomes singly-colored, the configuration becomes {\sf Conf BW-MR} or {\sf Conf RW-MB}.

Thus, the lemma holds.
\end{proof}

To show that the gathering is achieved, by Lemmas~\ref{RB3} and \ref{C89}, we consider the gathering only from {\sf Conf BW-MR} and {\sf Conf RW-MB} respectively.

\begin{lemma}[Lemma 7 in the main part]\label{Conf8}
From {\sf Conf BW-MR}, the gathering is achieved.
\end{lemma}
\begin{proof}
Let $u_i$ be the node occupied by Red robots and Blue robots with outdated views and Red robots without outdated views.
In $u_j$, after White robots execute R4b, the node becomes singly-colored with Blue robots.
In $u_i$, Blue robots with outdated views eventually move to $u_j$, after that, $u_i$ becomes singly-colored and robots can execute R2b during they can look White robots in $u_j$.
Thus, after every White robot changes their color, $u_j$ is a singly-colored node with Blue robots and there are Red robots and Blue robots with outdated views and Red robots without outdated views in $u_i$.
At that time, every robot in $u_j$ cannot execute any rule by the definition of the algorithm.
If there is no Red robot without an outdated view in $u_i$, then every robot in $u_i$ eventually moves to $u_j$ and the gathering is achieved.
If there are Red robots without outdated views in $u_i$, then every robot with an outdated view in $u_i$ eventually moves to $u_j$ and $u_j$ becomes a singly-colored node with Red robots.
After that, Red robots in $u_j$ can execute R5, and the gathering is achieved. Thus, the lemma holds.
\end{proof}

\begin{lemma}[Lemma 8 in the main part]\label{Conf9}
From {\sf Conf RW-MB}, the gathering is achieved.
\end{lemma}
\begin{proof}
Let $u_i$ be the node occupied by Red robots and Blue robots with outdated views and Blue robots without outdated views.
In $u_j$, after White robots execute R4a, the node becomes singly-colored with Red robots.
In $u_i$, Red robots with outdated views eventually move to $u_j$, after that, $u_i$ becomes singly-colored and robots can execute R3b during they can look White robots in $u_j$.
Thus, after every White robot changes their color, $u_j$ is a singly-colored node with Red robots and there are Red robots and Blue robots with outdated views and Blue robots without outdated views in $u_i$.
At that time, every robot in $u_j$ cannot execute any rule by the definition of the algorithm.

Consider the case that there are Blue robots without outdated views in $u_i$.
Then, Blue robots without outdated views in $u_i$ and Red robots in $u_j$ cannot execute any rule by the definition of the algorithm.
Every Red robot with an outdated view in $u_i$ eventually moves to $u_j$, $u_i$ becomes a singly-colored node with Blue robots. 
In $u_i$, some Blue robots are with outdated views and other Blue robots are without outdated views, but Blue robots without outdated views cannot execute any rules by the definition of the algorithm. 
In $u_j$, Red robots can execute R5.
After that, every Blue robot with an outdated view in $u_i$ eventually changes their color to Red and moves to $u_j$ and they can execute R5. Then, the gathering is achieved.
    
Consider the case that there is no Blue robot without an outdated view in $u_j$.
Then, Red robots in $u_j$ cannot execute any rule by the definition of the algorithm.
Every Red robot with an outdated view in $u_i$ eventually moves to $u_j$, and $u_i$ becomes a singly-colored node with Blue robots.
After that, Red robots in $u_j$ can execute R5 if they look at Blue robots with outdated views in $u_i$.
\begin{itemize}
    \item If every Blue robot with an outdated view in $u_i$ becomes Red before Red robots in $u_j$ look, then Red robots in $u_j$ cannot execute any rules until Red robots in $u_i$ move to $u_j$. Then, the gathering is achieved.
    \item If Red robots in $u_j$ look Blue robots with outdated views in $u_i$, they change their color to Blue and move to $u_i$. Then, they are Blue robots without outdated views in $u_i$, and we can apply the above discussion to this case. Then, the gathering is achieved.
\end{itemize}
Thus, the lemma holds.
\end{proof}

By Lemmas~\ref{Conf8} and \ref{Conf9}, we can derive the following theorem.
\begin{theorem}[Theorem 9 in the main part]
Gathering is solvable in full-light of 3 colors when $M_{init}$ is odd.
\end{theorem}

%% file: algorithm2-full.tex
In this case, we can assume $\phi>1$ because the visibility graph is connected.
The strategy of our algorithm is similar to Algorithm 1.
Initially, all robots are White, and robots on two border nodes become Red in their first activation.
The two border robots keep their lights Red or Blue, then the algorithm can recognize that they are originally border robots.
When non-border White robots become border robots, 
they change their color to Red (resp., Blue) if borders that join the node have Red (resp., Blue). 
To keep the connected visibility graph, when a border node becomes singly-colored, the border robot moves toward the center node.
At that time, if there exists a White robot in the directed neighboring node, the border robot changes its color.
Otherwise, it just moves without changing its color.
Eventually, two border nodes become neighboring.
In this algorithm, when two border nodes are neighboring, an additional color Purple is used to decide the gathering point.

\begin{algorithm}[p] 
\input{alg2-full.tex}

\caption{Algorithm for the case that $M_{init}$ is even and $O_{init}$ is odd.}
\label{alg2b}
\end{algorithm}

The rules of our algorithm are as follows:
\begin{itemize}
\item R0: If the gathering is achieved, a robot does nothing.
\item R1: A border White robot on a singly-colored border node changes its light to Red.
\item R2: A border Red robot on a singly-colored border node moves toward an occupied node without changing its color when there is no White robot on the neighboring node (R2a-1, R2a-2). A border Red robot moves toward an occupied node and changes its light to Blue only when there is at least one White robot on the neighboring node (R2b).
\item R3: A border Blue robot on a singly-colored border node moves toward an occupied node without changing its color when there is no White robot on the neighboring node (R3a-1, R3a-2). A border Blue robot moves toward an occupied node and changes its light to Red only when there is at least one White robot on the neighboring node (R3b).
\item R4: When White robots become border robots, they change their color to the same color as the border Red or Blue robots.
\item R5: If two border nodes are neighboring, every robot moves to the neighboring node with Purple robots (R5a).
A border Blue robot on a singly-colored border node changes its light to Purple when there are only Red robots or Red and Blue robots on the neighboring node (R5b-1, R5b-2).
A border Blue robot changes its light to Purple when there is Red robot on the same node and the neighboring node is a singly-colored node with Red robots (R5b-3).   
\end{itemize}

\begin{figure}[t]
\centering\includegraphics[keepaspectratio, width=11cm]{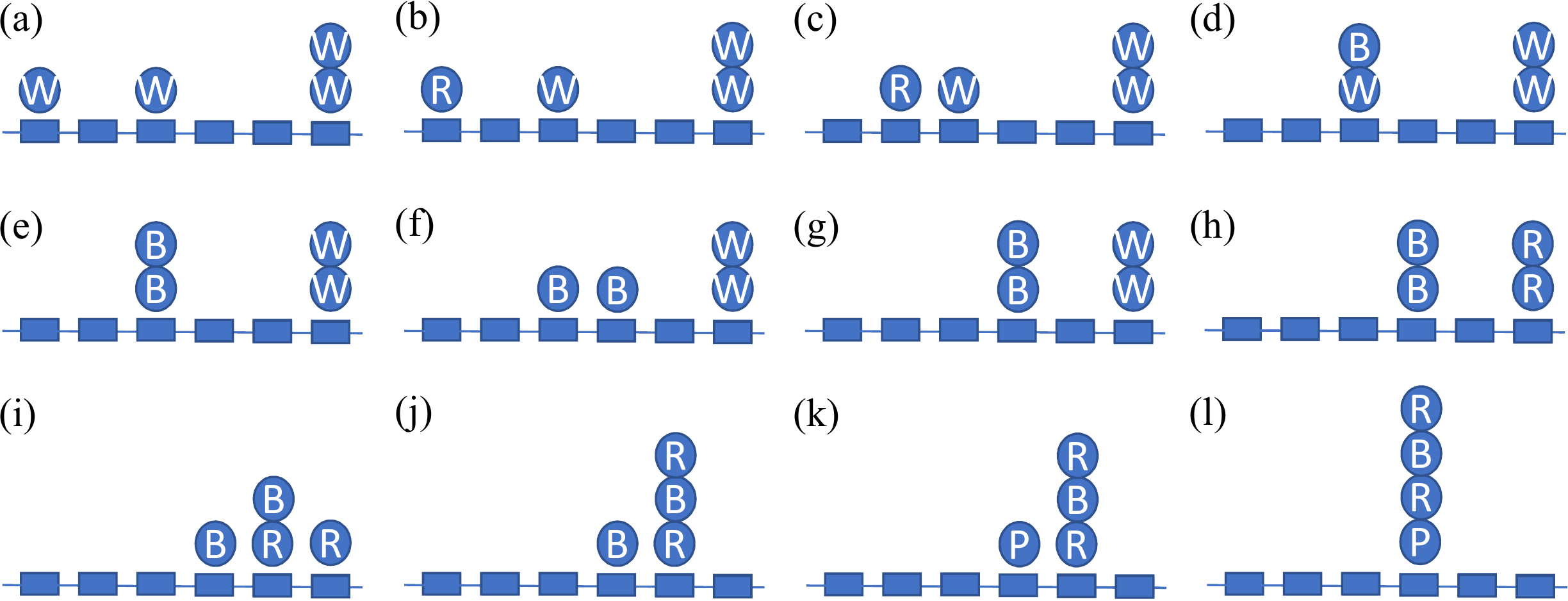}
\caption{An execution example of Algorithm \ref{alg2b} with $\phi=3$.}
\label{fig:ex-alg2}
\end{figure}

The formal description of the algorithm is in Algorithm~\ref{alg2b}.
Figure \ref{fig:ex-alg2} illustrates an execution example of Algorithm \ref{alg2b}. This figure assumes $\phi=3$. Figure \ref{fig:ex-alg2}(a) shows an initial configuration. First, the left border White robot changes its light to Red by R1 (Fig.\,\ref{fig:ex-alg2}(b)). Next, the left border Red robot moves by R2a-1 (Fig.\,\ref{fig:ex-alg2}(c)). Note that, here, the robot does not change its light. In the next movement, the left border Red robot moves to a node with a White robot by R2b, and hence it changes its light to Blue (Fig.\,\ref{fig:ex-alg2}(d)). Then the left border White robot changes its light to Blue (Fig.\,\ref{fig:ex-alg2}(e)). Left border Blue robots can move by R3a-1. In Fig.\,\ref{fig:ex-alg2}(f), one of them completes the movement. In this case, another robot can move by R3a-2 (Fig.\,\ref{fig:ex-alg2}(g)). Next, right border White robots change their lights by R1 (Fig.\,\ref{fig:ex-alg2}(h)). After that, left and right border robots move toward each other by R2a-1 and R3a-1. In Fig.\,\ref{fig:ex-alg2}(i), some Blue and Red robots meet at a node but border robots continue to move until the number of occupied nodes is at most two by R2a-2 and R3a-2 (Fig.\,\ref{fig:ex-alg2}(j)). After the number of occupied nodes is at most two, some robots change their lights to Purple. In this case, the left Blue robot changes its light by R5b-2 (Fig.\,\ref{fig:ex-alg2}(k)). After that, all robots move to the node with a Purple robot by R5a and achieve gathering (Fig.\,\ref{fig:ex-alg2}(l)).


%% file: alg2-full.tex
{\bf Colors}\\
W (White), R (Red), B (Blue), P (Purple)\\
\\
{\bf Rules}\\
/* Do nothing after gathering. */\\
R0: $\emptyset^\phi[?]\emptyset^\phi$ :: $\bot$\\

/* Start by the initial border robots. */\\
R1: $\emptyset^\phi[W!](\neg\emptyset^{\phi})$ :: $R$\\

/* Border robots on singly-colored nodes move inwards. */\\
R2a-1: $\emptyset^\phi[R!](\emptyset)(\neg(\emptyset^{\phi-1}))$ :: $\rightarrow$\\

R2a-2: $\emptyset^\phi[R!]\begin{pmatrix}
\neg W\\
R
\end{pmatrix}(\neg(\emptyset^{\phi-1}))$ :: $\rightarrow$\\

R2b: $\emptyset^\phi[R!](W)(?^{\phi-1})$ :: $B, \rightarrow$\\

R3a-1: $\emptyset^\phi[B!](\emptyset)(\neg(\emptyset^{\phi-1}))$ :: $\rightarrow$\\

R3a-2: $\emptyset^\phi[B!]\begin{pmatrix}
\neg W\\
B
\end{pmatrix}(\neg(\emptyset^{\phi-1}))$ :: $\rightarrow$\\

R3b: $\emptyset^\phi[B!](W)(?^{\phi-1})$ :: $R, \rightarrow$\\

/* When White robots become border robots, they change their color to the same color as the border robots. */\\
R4a: $\emptyset^\phi\begin{bmatrix}
R\\
[W]
\end{bmatrix}(\neg\emptyset^{\phi})$ :: $R$\\

R4b: $\emptyset^\phi\begin{bmatrix}
B\\
[W]
\end{bmatrix}(\neg\emptyset^{\phi})$ :: $B$\\

/* When two border nodes are neighboring, they gather to the border node with Purple robots. */\\
R5a: $\emptyset^\phi[?](P)(\emptyset^{\phi-1})$ :: $\rightarrow$\\

R5b-1: $\emptyset^\phi[B!](R!)(\emptyset^{\phi-1})$ :: $P$\\


R5b-2: $\emptyset^{\phi}[B!]\begin{pmatrix}
R\\
B
\end{pmatrix}(\emptyset^{\phi-1})$ :: $P$\\


R5b-3: $\emptyset^\phi\begin{bmatrix}
R\\
[B]
\end{bmatrix}(R!)(\emptyset^{\phi-1})$ :: $P$

%% file: alg2-correctness-full.tex
Just the same as Algorithm 1, since there is no rule that non-border robot can execute by the definition of Algorithm 2, the following lemma holds.

\begin{lemma}[Lemma 10 in the main part]\label{NB2}
When a robot $r_i$ looks, if it is a non-border, it cannot execute any action.
\end{lemma}

To discuss the correctness, we change the definition of mega-cycle as follows:
We consider the time instants such that the number of occupied nodes with White robots among non-border nodes (denoted as $\#O_{W}$) has just reduced at least one. That is, mega-cycles end when either border position moves to the nearest occupied node with White robots.
If a border position moves to the nearest occupied node with White robots, we say \emph{the border absorbs White robots}. 


Figure~\ref{TD2} shows a transition diagram of configurations for every mega-cycle.
The doubly (resp. singly) lined arrows represent that $\#O_{W}$ is decreased by 2 (resp. 1).
In the diagram, {\sf Init} and {\sf Conf 1-7} are the same as those of Algorithm 1 and they have the same transitions between them as shown in Figure~\ref{TD}, where note that each node with W circle contains at least one White robot. In addition to these configurations, there exist four configurations {\sf Conf 8-11} in Algorithm 2. 
In {\sf Conf 1-3}, both borders absorb White robots.
In {\sf Conf 4-7}, when one border absorbs White robots, the other border is neighbored to an occupied node with a White robot. 
On the other hand, in {\sf Conf 8-11}, when one border absorbs White robots, the other border is not neighbored to any occupied node with a White robot.


\begin{figure*}[t]
\centering\includegraphics[keepaspectratio, width=13cm]{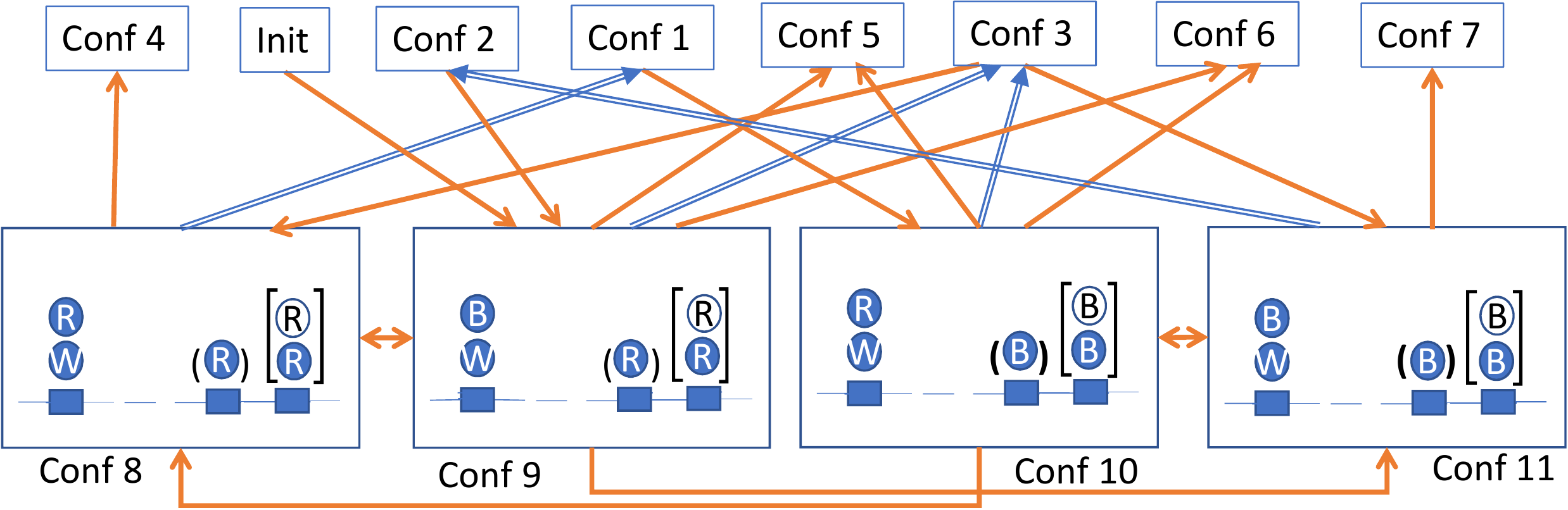}
\caption{Transition Diagram for Algorithm 2 while $D>1$.}\label{TD2}
\end{figure*}

The following lemma can be proved similarly to the proof of Lemma~\ref{Reduction-by-one}.

\begin{lemma}[Lemma 11 in the main part]\label{Reduction-by-one-for-alg-2}
Assume that a border node $u$ is singly-colored with Red (resp. Blue) at time $t$ and there is no robot in $u$ at time $t' (> t)$.
If the neighboring node of $u$ (denoted as $v$) is an occupied node with a White robot at $t$, 
all robots in $u$ change their color to Blue (resp. Red), move to $v$, and $\#O_{W}$ and $D$
are reduced by at least one at $t'$.
Otherwise, that is, when $v$ is a node without White robots at $t$,
all robots move to $v$ and do not change their color and $D$
is reduced by at least one at $t'$.
\end{lemma}

By Lemma~\ref{Reduction-by-one-for-alg-2}, each border node moves to the occupied node until either border absorbs White robots in any mega-cycle.  
In the following lemmas, transitions among {\sf Init} and {\sf Conf 1-7} can be proved similarly to those in the corresponding lemmas using Lemma~\ref{Reduction-by-one-for-alg-2} instead of Lemma~\ref{Reduction-by-one}.
The difference occurs when one border absorbs White robots and the neighboring node of the other border is a node without White robots. In this case, since the other border moves to the neighboring node, the configuration becomes {\sf Conf 8-11}. The proofs can be done similarly.

\begin{lemma}\label{Init2}
{\sf Init} $\xrightarrow{MC}$ {\sf Conf 1},  {\sf Init} $\xrightarrow{MC}$ {\sf Conf 3},  {\sf Init} $\xrightarrow{MC}$ {\sf Conf 4}, or {\sf Init} $\xrightarrow{MC}$ {\sf Conf 9}.
\end{lemma}

\begin{lemma}\label{Conf12}
{\sf Conf 1} $\xrightarrow{MC}$ {\sf Conf 2},  {\sf Conf 1} $\xrightarrow{MC}$ {\sf Conf 3}, {\sf Conf 1} $\xrightarrow{MC}$ {\sf Conf 7}, or {\sf Conf 1} $\xrightarrow{MC}$ {\sf Conf 10}.
\end{lemma}

\begin{lemma}\label{Conf22}
{\sf Conf 2} $\xrightarrow{MC}$ {\sf Conf 1},  {\sf Conf 2} $\xrightarrow{MC}$ {\sf Conf 3}, {\sf Conf 2} $\xrightarrow{MC}$ {\sf Conf 4} or
 {\sf Conf 2} $\xrightarrow{MC}$ {\sf Conf 9}.
\end{lemma}

\begin{lemma}\label{Conf32}
{\sf Conf 3} $\xrightarrow{MC}$ {\sf Conf 1},  {\sf Conf 3} $\xrightarrow{MC}$ {\sf Conf 2}, 
 {\sf Conf 3} $\xrightarrow{MC}$ {\sf Conf 3},
{\sf Conf 3} $\xrightarrow{MC}$ {\sf Conf 5}, {\sf Conf 3} $\xrightarrow{MC}$ {\sf Conf 6}, 
{\sf Conf 3} $\xrightarrow{MC}$ {\sf Conf 8}, 
or 
{\sf Conf 3} $\xrightarrow{MC}$ {\sf Conf 11}.

\end{lemma}

\begin{lemma}\label{Conf42}
{\sf Conf 4} $\xrightarrow{MC}$ {\sf Conf 1},  {\sf Conf 4} $\xrightarrow{MC}$ {\sf Conf 3}, 
 {\sf Conf 4} $\xrightarrow{MC}$ {\sf Conf 5},
or {\sf Conf 4} $\xrightarrow{MC}$ {\sf Conf 6}.
\end{lemma}

\begin{lemma}\label{Conf72}
{\sf Conf 7} $\xrightarrow{MC}$ {\sf Conf 2},  {\sf Conf 7} $\xrightarrow{MC}$ {\sf Conf 3}, 
 {\sf Conf 7} $\xrightarrow{MC}$ {\sf Conf 5},
\\or {\sf Conf 7} $\xrightarrow{MC}$
{\sf Conf 6}.
\end{lemma}

\begin{lemma}\label{Conf52}
{\sf Conf 5} $\xrightarrow{MC}$ {\sf Conf 3},  {\sf Conf 5} $\xrightarrow{MC}$ {\sf Conf 4}, 
or {\sf Conf 5} $\xrightarrow{MC}$ {\sf Conf 1}.
\end{lemma}

\begin{lemma}\label{Conf62}
{\sf Conf 6} $\xrightarrow{MC}$ {\sf Conf 2},  {\sf Conf 6} $\xrightarrow{MC}$ {\sf Conf 3}, 
or {\sf Conf 6} $\xrightarrow{MC}$ {\sf Conf 7}.
\end{lemma}

The following lemmas treat transitions from {\sf Conf 8-11} and these proofs can be shown similarly.

\begin{lemma}\label{Conf82}
{\sf Conf 8} $\xrightarrow{MC}$ {\sf Conf 9},  {\sf Conf 8} $\xrightarrow{MC}$ {\sf Conf 1}, 
or {\sf Conf 8} $\xrightarrow{MC}$ {\sf Conf 4}.
\end{lemma}

\begin{lemma}\label{Conf92}
{\sf Conf 9} $\xrightarrow{MC}$ {\sf Conf 8},  {\sf Conf 9} $\xrightarrow{MC}$ {\sf Conf 11}, 
 {\sf Conf 9} $\xrightarrow{MC}$ {\sf Conf 3},
{\sf Conf 9} $\xrightarrow{MC}$ {\sf Conf 5},
or {\sf Conf 9} $\xrightarrow{MC}$ {\sf Conf 6}.
\end{lemma}

\begin{lemma}\label{Conf102}
{\sf Conf 10} $\xrightarrow{MC}$ {\sf Conf 8}, 
{\sf Conf 10} $\xrightarrow{MC}$ {\sf Conf 11},  {\sf Conf 10} $\xrightarrow{MC}$ {\sf Conf 3}, 
{\sf Conf 10} $\xrightarrow{MC}$ {\sf Conf 5},
or {\sf Conf 10} $\xrightarrow{MC}$ {\sf Conf 6}.
\end{lemma}

\begin{lemma}\label{Conf112}
{\sf Conf 11} $\xrightarrow{MC}$ {\sf Conf 10},  {\sf Conf 11} $\xrightarrow{MC}$ {\sf Conf 2}, 
or {\sf Conf 11} $\xrightarrow{MC}$ {\sf Conf 7}.
\end{lemma}

\begin{lemma}[Lemma 12 in the main part]\label{RB32}
From the initial configuration, $\#O_{W}$ 
decreases monotonically and eventually becomes at most one.
\end{lemma}
\begin{proof}
By Lemmas~\ref{Init2}-\ref{Conf112}, in each mega-cycle, the number of occupied nodes with White robots between two borders decreases by at least one.
Thus, the lemma holds.
\end{proof}


\begin{lemma}[Lemma 13 in the main part]\label{EO2}
Let $h$ be the number of occupied nodes where an original border robot $r_h$ absorbed White robots in $\mathcal{G}^\prime$ from the initial configuration and let $u_h$ denote the current border node $r_h$ is located.    
If $h$ is odd (resp. even), $r_h$'s light is Blue (resp. Red) when $r_h$ comes into $u_h$.
\end{lemma}

It can be easily verified by Figures~\ref{TD} and~\ref{TD2} (Lemmas~\ref{Init2}-\ref{Conf112}) and Lemmas~\ref{RB32}--\ref{EO2} that the following configurations occur 
when 
$\#O_{W}$ becomes at most one.

\begin{enumerate}
    \item {\sf Conf 1}, {\sf Conf 2}, {\sf Conf 5} and {\sf Conf 6} ($D \geq 2$ and $\#O_{W}=1$).
    \item  {\sf Conf 3} ($D \geq 1$ and $\#O_{W}=0$) and {\sf Conf 9} ($D \geq 2$ and $\#O_{W}=0$), and {\sf Conf 10} ($D \geq 2$ and $\#O_{W}=0$).
\end{enumerate}

In the former case, for configurations {\sf Conf 1,2,5,6}($D=2$ and $\#O_{W}=1$), we have the following lemma, where {\sf Conf RW-MB} and {\sf Conf BW-MR} have been defined in the proof for Algorithm~1.
In {\sf Conf 3}($D=1$), both border nodes
may contain White robots with outdated views
and contain no White robots (Figure~\ref{fig:RWBW}).

\begin{figure}
    \centering
    \begin{minipage}[t]{0.2\linewidth}\centering
    \includegraphics[keepaspectratio, scale=0.15]{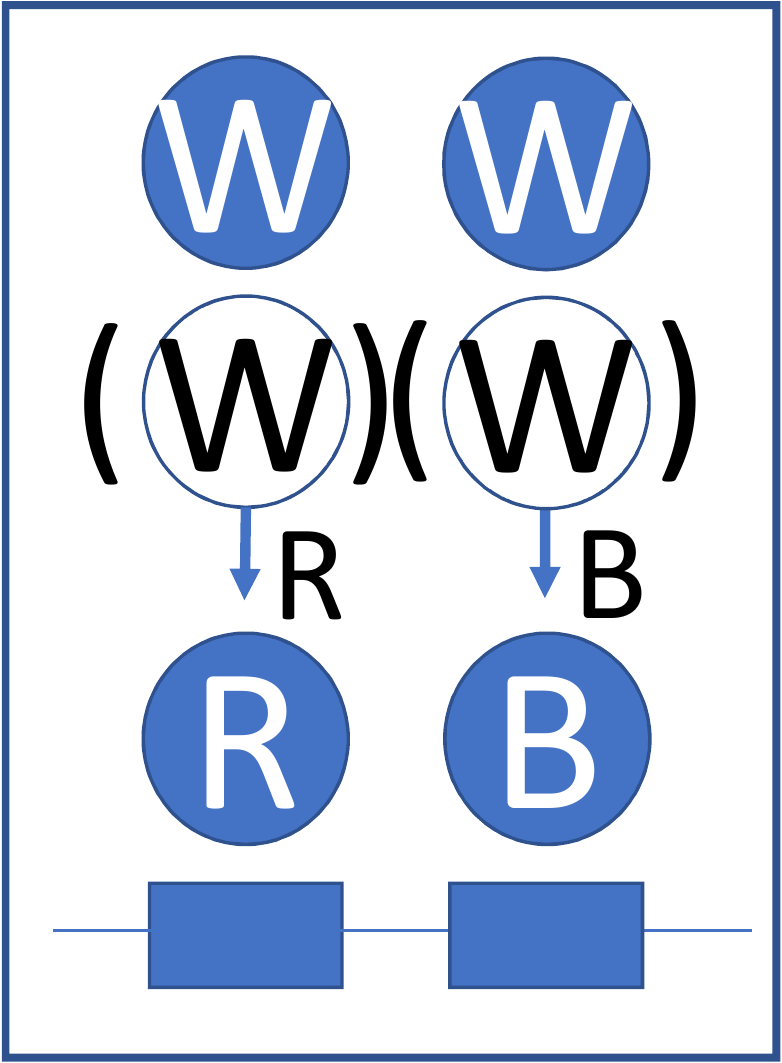} 
    \end{minipage}
    \begin{minipage}[t]{0.2\linewidth}\centering
    \includegraphics[keepaspectratio, scale=0.15]{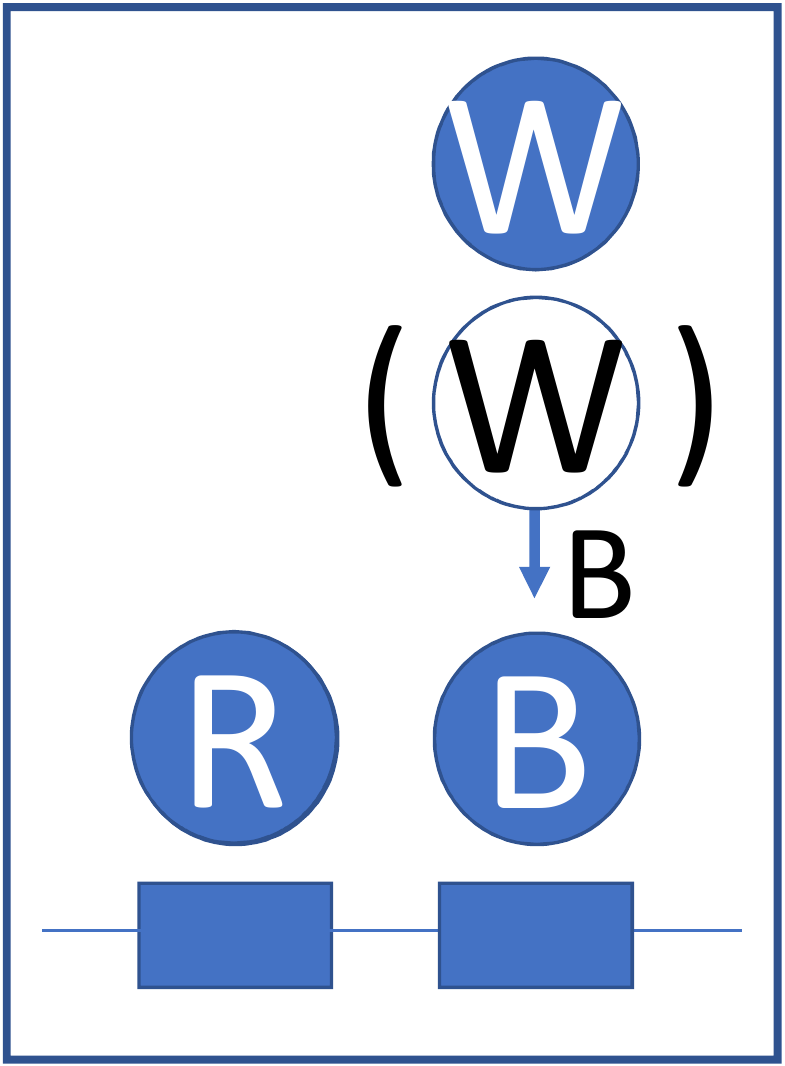} 
    \end{minipage}
    \begin{minipage}[t]{0.2\linewidth}\centering
    \includegraphics[keepaspectratio, scale=0.15]{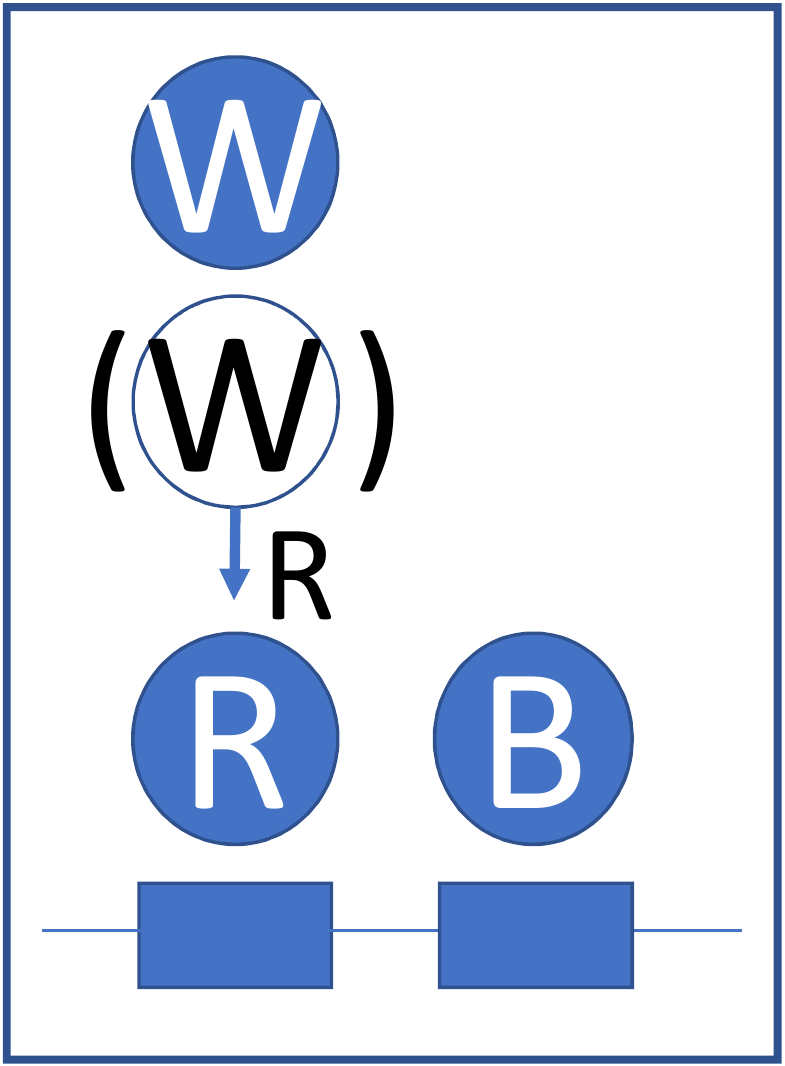}
    \end{minipage}
    \begin{minipage}[t]{0.2\linewidth}\centering
    \includegraphics[keepaspectratio, scale=0.15]{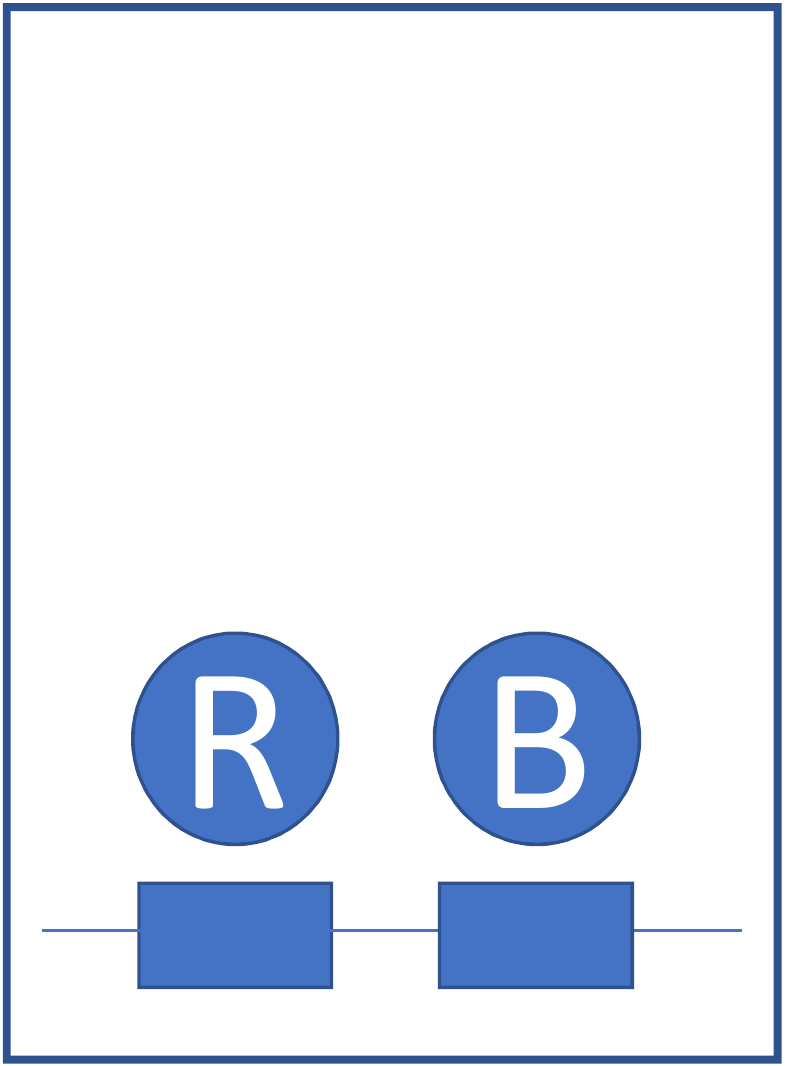}
    \end{minipage}
    \caption{The patterns of the configurations {\sf Conf 3}($D=1$).}
    \label{fig:RWBW}
\end{figure}

\begin{lemma}[Lemma 14 in the main part]\label{Conf1256andDeq2}
\begin{enumerate}
    \item {\sf Conf 1}($D=2$ and $\#O_{W}=1$) $\rightarrow^*$ {\sf Conf RW-MB}, or {\sf Conf 3}($D=1$)
    \item {\sf Conf 2}($D=2$ and $\#O_{W}=1$) $\rightarrow^*$ {\sf Conf BW-MR}, or {\sf Conf 3}($D=1$)
    \item {\sf Conf 5}($D=2$ and $\#O_{W}=1$) $\rightarrow^*$ {\sf Conf BW-MR}, or {\sf Conf 3}($D=1$)
    \item {\sf Conf 6}($D=2$ and $\#O_{W}=1$) $\rightarrow^*$ {\sf Conf RW-MB}, or {\sf Conf 3}($D=1$)
\end{enumerate}
\end{lemma}
\begin{proof}
{\bf Case 1}({\sf Conf 1}($D=2$ and $\#O_{W}=1$)).
Let $u_i$ and $u_k$ be nodes occupied by Blue and White robots and let $u_j$ be the node occupied by White robots neighboring to $u_i$ and $u_k$.
In $u_i$ and $u_k$, only White robots execute R4b and change their color to Blue, and then, when there are no White robots, Blue robots will execute R3b and change their color to Red and move to $u_j$.
Let $t$ be the first time such that $u_i$ or $u_k$ has no White robots. 

If both $u_i$ and $u_k$ have no White robots at $t$, $u_i$ and $u_k$ contain only Blue robots and they execute R3b, change their color to Red and move to $u_j$. 
Thus, when the distance between the borders becomes one, the configuration becomes {\sf Conf RW-MB}.
Otherwise, without loss of generality, $u_i$ contains only Blue robots and $u_k$ contains Blue and White robots without outdated views and White robots with outdated views. 
Since Blue robots in $u_i$ execute R3b, they change their color to Red and move to $u_j$.
Since White robots in $u_k$ execute R4b, after $u_k$ contains no White robots, Blue robots in $u_k$ execute R3b, change their color to Red and move to $u_j$.
When the distance between the borders becomes one, the configuration becomes {\sf Conf 3}($D=1$)
if all robots in $u_i$ move to $u_j$ earlier, and becomes {\sf Conf RW-MB} if all robots in $u_k$ move to $u_j$ earlier. \\

\noindent
{\bf Case 2}({\sf Conf 2}($D=2$ and $\#O_{W}=1$)), {\bf Case 3}({\sf Conf 5}($D=2$ and $\#O_{W}=1$)) and {\bf Case 4}({\sf Conf 6}($D=2$ and $\#O_{W}=1$)) can be proved similarly to {\bf Case 1}.


\end{proof}


Otherwise ($D \geq 3$), we have the following transitions.

\begin{lemma}[Lemma 15 in the main part]\label{Conf1256Dgeq3}
\begin{enumerate}
    \item {\sf Conf 1}($D \geq 3$ and $\#O_{W}=1$) $\xrightarrow{MC}^*$ {\sf Conf 3}($D \geq 1$ and $\#O_{W}=0$), or {\sf Conf 10}($D \geq 2$ and $\#O_{W}=0$)
    \item {\sf Conf 2}($D \geq 3$ and $\#O_{W}=1$)$\xrightarrow{MC}^*$ {\sf Conf 3}($D\geq 1$ and $\#O_{W}=0$), or {\sf Conf 9}($D \geq 2$ and $\#O_{W}=0$)
    \item {\sf Conf 5}($D \geq 3$ and $\#O_{W}=1$)$\xrightarrow{MC}^*$ {\sf Conf 3}($D\geq 1$ and $\#O_{W}=0$), or {\sf Conf BW-MR}($D=1$ and $\#O_{W}=0$)
    \item {\sf Conf 6}($D \geq 3$ and $\#O_{W}=1$)$\xrightarrow{MC}^*$ {\sf Conf 3}($D\geq 1$ and $\#O_{W}=0$), or {\sf Conf RW-MB}($D=1$ and $\#O_{W}=0$)
\end{enumerate}
\end{lemma}

We can prove that  {\sf Conf BW-MR}, {\sf Conf RW-MB} and {\sf Conf 3}($D=1$) 
become gathering configuration in the following lemma.

\begin{lemma}[Lemma 16 in the main part]\label{Conf-RMB-BMR}
From configurations {\sf Conf BW-MR}, {\sf Conf RW-MB} and {\sf Conf 3}($D=1$), 
gathering is achieved.
\end{lemma}
\begin{proof}
{\bf Case 1:}({\sf Conf BW-MR})
Let $u_i$ be the node occupied by Red robots and Blue robots with outdated view and Red robots without outdated views and let $u_j$ be the node occupied by White and Blue robots without outdated views.
In $u_j$, White robots execute R4b and change their color to Blue. At the same time, Red robots in $u_i$ execute R2b, change their color to Blue and move to $u_j$.

If all robots in $u_i$ have executed R2b before White robots do not exist in $u_j$, gathering is achieved in $u_j$ because Blue robots in $u_j$ cannot execute any rule and White robots in $u_j$ do not move.
Otherwise, that is, there do not exist White robots in $u_j$ before all robots in $u_i$ move by R2b, Blue robots in $u_j$ execute R5b-2 (if $u_i$ contains Blue and Red robots), R5b-1 (if $u_i$ contains only Red robots), or no rules (if $u_i$ contains only Blue robots). In the first and second cases, Blue robots in $u_j$ change their color to Purple, Red robots without outdated views execute R5a and move to $u_j$, and other robots with outdated views move to $u_j$. And
in the third case, all Blue robots in $u_i$ are with outdated views and will move to $u_j$. Thus, gathering is achieved.

\vspace{0.3\baselineskip}
\noindent{\bf Case 2:}({\sf Conf RW-MB}) This case can be proved similarly to Case 1. 

\vspace{0.3\baselineskip}
\noindent{\bf Case 3:}({\sf Conf 3}($D=1$)) 
Let $u_i$ be the node occupied by Red robots without outdated views and White robots with and/or without outdated views, and let $u_j$ be the node occupied by Blue robots without outdated views and White robots with and/or without outdated views.

White robots in $u_i$ (resp. $u_j$) execute R4a (resp. R4b) and change their color to Red (resp. Blue). Then since $u_i$ or $u_j$ does not contain White robots, let $t$ be the first time such that either $u_i$ or $u_j$ does not have White robots.
If both borders do not have White robots at time $t$, $u_i$ contains only Red robots and $u_j$ contains only Blue robots at $t$\footnote{In this configuration, both borders contain no White robots}. Thus, Blue robots change their color to Purple by R5b-1, and all Red robots in $u_i$ execute R5a and move to $u_j$, and gathering achieved.
In the case that there exist White robots in $u_i$ or $u_j$, these configurations become the same as those in {\bf Case 1} and therefore, will become gathering configurations. 

In $u_j$, White robots execute R4b and change their color to Red. At the same time, Red robots in $u_i$ execute R2b, change their color to Blue and move to $u_i$.
\end{proof}

Then by Lemmas~\ref{Conf1256andDeq2}-\ref{Conf-RMB-BMR}, it is sufficient to consider the configurations {\sf Conf 3}($D\geq 2$ and $\#O_{W}=0$), {\sf Conf 9}($D \geq 2$ and $\#O_{W}=0$), and {\sf Conf 10}($D \geq 2$ and $\#O_{W}=0$) for the former case.

The latter case has the following transitions. Note that, these transitions do not reduce $\#O_{W}$ and just reduces the distance between the two borders.
Note also that, the destinations of these transitions do not contain any White robots.

\begin{lemma}[Lemma 17 in the main part]\label{Conf3Deq2}
\begin{enumerate}
    \item {\sf Conf 3} ($D \geq 3$ and $\#O_{W}=0$) 
    $\rightarrow^*$ {\sf Conf 3}(($D=2$ and $\#O_{W}=0$) or {\sf Conf 3}($D=1$)
    \item {\sf Conf 3}($D \geq 3$ and $\#O_{W}=0$)$\rightarrow^*$ {\sf Conf 10}($D=2$ and $\#O_{W}=0$)
    \item {\sf Conf 3}($D \geq 3$ and $\#O_{W}=0$)$\rightarrow^*$ {\sf Conf 9}($D=2$ and $\#O_{W}=0$)
\end{enumerate}
\end{lemma}

Since {\sf Conf 9}($D \geq 2$ and $\#O_{W}=0$)(resp.   {\sf Conf 10}($D \geq 2$ and $\#O_{W}=0$)) becomes {\sf Conf 3}(($D=2$ and $\#O_{W}=0$) or {\sf Conf 3}($D=1$), or   {\sf Conf 9}($D=2$ and $\#O_{W}=0$) (resp. {\sf Conf 9}($D=2$ and $\#O_{W}=0$)), the correctness proof completes if we can show that these three configurations become gathering ones.

\begin{lemma}[Lemma 18 in the main part]\label{conf3-9-10}
{\sf Conf 3}($D=2$ and $\#O_{W}=0$),  {\sf Conf 9}($D=2$ and $\#O_{W}=0$), and {\sf Conf 10}($D=2$ and $\#O_{W}=0$) become gathering configurations.
\end{lemma}
\begin{proof}
{\bf Case-1}({\sf Conf 3}($D=2$)):
Let $u_i$ and $u_k$ be the nodes occupied by Red and White robots and Blue and White robots, respectively, and let $u_j$ is an empty neighboring node to $u_i$ and $u_k$.
In $u_i$ (resp. $u_k$), White robots execute R4a (resp. R4b) and change their color to Red (resp. Blue), and then, when there is no White robot, Red robots (resp. Blue robots) 
execute R2a-1 (resp. R3a-1) and move to $u_j$ without changing their color. Then $u_j$ contains Red and Blue robots.
Since R5a and R5b cannot apply to configurations with $D=2$, the distance of the two borders eventually becomes one.
Let $t$ be a time such that all robots in either border node move to $u_j$. If all robots in both borders move to $u_j$ at time $t$,
gathering is achieved.
Otherwise, there are two cases,
{\bf (Case 1-1)} all robots in $u_i$ move to
$u_j$ and {\bf (Case 1-2)} all robots in $u_k$ move to $u_j$.

{\bf (Case 1-1)}: 
%
%
If all robots in $u_i$ move to $u_j$ at $t$, $u_j$ contains all Red robots located in $u_i$. In this case, we can consider the following two subcases:
\begin{itemize}\setlength{\parskip}{0cm}\setlength{\itemsep}{0cm}
\item If $u_j$ contains some Blue robots located in $u_k$ at $t$, they have non-outdated views and there are only Blue robots in $u_k$. Then Blue robots in $u_k$ can execute R5b-2 and change their color to Purple, and all robots in $u_j$ execute R5a and move to $u_k$, gathering is achieved.
\item Otherwise, that is, $u_j$ contains only Red robots, and there are Blue and (possibly empty) White robots in $u_k$, containing White ones with outdated views\footnote{These robots will only change their color to Blue in $u_k$.}.
Since $u_k$ contains Blue robots and (possibly empty) White robots, the configuration is {\sf Conf 3}($D=1$). Then gathering is achieved by Lemma~\ref{Conf-RMB-BMR}. 
\end{itemize} 

\vspace{-0.2\baselineskip}
{\bf (Case 1-2)}: This case can be proved similarly.

\vspace{0.3\baselineskip}
\noindent
{\bf Case-2}({\sf Conf 9}($D=2$ and $\#O_{W}=0$)): 
Let $u_i$ be nodes occupied by Blue and White robots without outdated views, let $u_k$ be nodes occupied by Red robots with and without outdated views, and let $u_j$ is a node occupied by Red robots without outdated views or an empty node neighboring to $u_i$ and $u_k$.
In $u_i$, White robots execute R4b and change their color to Blue, and then, when there is no White robot, Blue robots 
execute R3a-2 and move to $u_j$ without changing their color. Then $u_j$ contains Red and Blue robots.
In $u_k$, Red robots without outdated views execute R2a-2 and move to $u_j$ without changing their color. Since R5a and R5b cannot apply to configurations with $D=2$, the distance of the two borders eventually becomes one.
Let $t$ be a time such that all robots in either border node move to $u_j$. If all robots in both borders move to $u_j$ at time $t$,
gathering is achieved.
Otherwise, there are two cases,
{\bf (Case 2-1)} all robots in $u_i$ move to
$u_j$ and {\bf (Case 2-2)} all robots in $u_k$ move to $u_j$.

{\bf (Case 2-1)}: 
If all robots in $u_i$ move to $u_j$ at $t$, $u_j$ contains all robots located in $u_i$ and some (possibly empty) Red robots located in $u_k$ and they have non-outdated views. The border node $u_k$ contains Red robots with and without outdated views at $t$. Then Blue robots in $u_j$ execute R5b-3 and change their color to Purple, and Red robots without outdated views in $u_k$
move to $u_j$ by R5a or Red robots with outdated views are moving to $u_j$. Thus gathering is achieved.

{\bf (Case 2-2)}: 
If all Red robots in $u_k$ move to $u_j$ at $t$, $u_j$ contains all Red robots located in $u_k$.
In this case, we can consider the following two subcases:
\begin{itemize}\setlength{\parskip}{0cm}\setlength{\itemsep}{0cm}
\item If $u_j$ contains some Blue robots located in $u_i$ at $t$, they have non-outdated views and there are only Blue robots on $u_i$. Then, Blue robots in $u_i$ can execute R5b-2 and change their color to Purple. Thus, all robots gather in $u_i$ by R5a.
\item Otherwise, that is, $u_j$ contains only Red robots, and there are Blue and (possibly empty) White robots in $u_i$, containing White ones with outdated views. Then, the configuration is {\sf Conf 3}($D=1$). Then, the gathering is achieved by Lemma~\ref{Conf-RMB-BMR}.
\end{itemize}

\noindent
{\bf Case-3}({\sf Conf 10}($D=2$ and $\#O_{W}=0$)) can be proved similarly to Case 2.

\end{proof}
%

By the above discussion, we can derive the following theorem.
\begin{theorem}[Theorem 19 in the main part]
Gathering is solvable in full-light of 4 colors when $M_{init}$ is even and $O_{init}$ is odd.
\end{theorem}

It is an interesting open question whether gathering is solvable or not in full-light of 3 colors when $M_{init}$ is even and $O_{init}$ is odd.

%% file: impossibility.tex
First, we consider the case that $M_{init}$ and $O_{init}$ are even.

\begin{definition}[Definition 20 in the main part]
A configuration is \emph{edge-view-symmetric} if there exist at least two distinct nodes hosting each at least one robot, and an edge $(u_i,u_{i+1})$ such that, for any integer $k\geq0$, and for any robot $r_1$ at node $u_{i-k}$, there exists a robot $r_2$ at node $u_{i+k+1}$ such that $\mathcal{V}_1=\mathcal{V}_2$.
\end{definition}

\begin{theorem}[Theorem 21 in the main part]
\label{thm:edge-view-symmetry}
Deterministic gathering is impossible from any edge-view-symmetric configuration.
\end{theorem}

\begin{proof}
Let us first observe that a gathered configuration is \emph{not} edge-view-symmetric (by definition).

Now, we show that starting from any edge-view-symmetric configuration, the scheduler can preserve an edge-view-symmetric configuration forever.
Suppose we start from an edge-view-symmetric configuration for edge $(u_i,u_{i+1})$.
Anytime a robot $r_1$ at node $u_{i-k}$ is enabled, for some integer $k\geq0$, execute the Look phase of all robots at $u_{i-k}$ with color $L_1$ (all those robots have the same view as $r_1$, and thus obtain the same information), and the Look phase of all robots at node $u_{i+k+1}$ with color $L_1$ (all those robots have the same view as $r_1$, and thus obtain the same information). Now, the scheduler executes the Compute phase of all aforementioned robots (they thus obtain the same (possibly new) color and the same move decision). Last, execute the Move phase of all those robots, since their move decision was the same in the Compute phase. 
The scheduler can remain fair by executing robots in a double round robin order (first by hosting node, second by robot color), yet, the execution contains only edge-view-symmetric configurations, hence never reaches gathering. 
\end{proof}

\begin{corollary}[Corollary 22 in the main part]
\label{cor:imp-ee}
Starting from a configuration where $M_{init}$ is even and $O_{init}$ is even, and all robots have the same color, deterministic gathering is impossible.
\end{corollary}

\begin{proof}
When $M_{init}$ is even and $O_{init}$ is even, if all robots initially share the same color, it is possible to construct an edge-view-symmetric initial configuration, from which deterministic gathering is impossible.
\end{proof}

\begin{corollary}
Starting from a configuration where $M_{init}$ is even and $O_{init}$ is even, and all initial colors are shared by at least two robots, there exist initial configurations (e.g. edge-view-symmetric configurations) that a deterministic algorithm cannot gather.
\end{corollary}

\begin{proof}
When $M_{init}$ is even and $O_{init}$ is even, if all colors are initially shared by at least two robots, it is possible to construct an edge-view-symmetric initial configuration, from which deterministic gathering is impossible.
\end{proof}

Next, we consider the case that $O_{init}=2$ and $K=1$.
\begin{corollary}
If the initial configuration contains only two occupied positions and robots never change their colors, gathering is impossible.
\end{corollary}

\begin{proof}
Suppose for the purpose of contradiction that there exists a gathering algorithm from an initial configuration with two occupied positions and all robots have the same color.
By Theorem~\ref{thm:edge-view-symmetry}, this implies that the configuration is \emph{not} edge-view-symmetric, yet the number of nodes between the two locations (i.e., $M_{init}$) is odd. So, the robots at both occupied locations execute exactly the same algorithm when activated by the scheduler. 

Since by assumption, robots never change color, they can either move or not move, and if they move they may move toward the other location or further from the other location. If robots don't move, then the reached configuration is the same as the initial one, where gathering is not achieved, hence the assumption is false.
If robots move away, then the scheduler activates only robots at one location, as a result, the reached configuration is edge-view-symmetric. 
Similarly, if robots move toward the other location, then the scheduler activates only robots at one location, and the reached configuration is edge-view-symmetric. 
Overall, a contradiction, hence the corollary holds.
\end{proof}

We now study the impact of an important property our algorithms satisfy: cautiousness~\cite{BPT10j}.
\begin{definition}[Definition 23 in the main part]
A gathering algorithm is \emph{cautious} if, in any execution, the direction to move is only toward other occupied nodes, i.e.,
robots are not adventurous and do not want to expand the covered area.
\end{definition}

Note that the algorithms we provide in previous sections, only border robots move, and they only move toward occupied other nodes, hence our algorithms are cautious.

\begin{lemma}[Lemma 24 in the main part]
\label{lem:morethantwoborders}
A cautious algorithm that starts from a configuration with more than two borders cannot solve gathering. 
\end{lemma}
\begin{proof}
Since robots are not aware of $R$, the total number of robots, the algorithm must work irrespective of $R$.
As there can be only an even number of borders in a configuration (all robots have the same visibility range), having more than two borders implies having at least two distinct parts $\mathcal{A}$ and $\mathcal{B}$, separated in each of the two directions of the ring by at least $k>\phi$ empty nodes (That is, $H_{init}>\phi$).

Since the algorithm is cautious, the robots in $\mathcal{A}$ must occupy positions that are within the borders of $\mathcal{A}$ in the remaining of the execution. Similarly, the robots in $\mathcal{B}$ must occupy positions that within the borders of $\mathcal{B}$ in the remaining of the execution. As a result, robots in $\mathcal{A}$ and in $\emph{B}$ never merge, and thus gathering is not achieved. 
\end{proof}

\begin{lemma}[Lemma 25 in the main part]
\label{lem:noborder}
A cautious algorithm that starts from a configuration with no border cannot solve gathering. 
\end{lemma}
\begin{proof}
Since robots are not aware of $R$, the total number of robots, the algorithm must work irrespective of $R$.
Suppose for the purpose of contradiction, that there exists a gathering algorithm for an initial configuration with no border. If the algorithm never creates a hole (and hence two borders), gathering is never achieved. Hence, there exists a step in the execution that creates two borders from a given configuration $c_i$ by the move of robot $r_j$. Now consider configuration $c_i'$ defined as follows: the new ring is twice as big, configuration $c_i$ is repeated twice on the new ring (assuming robot $r_j$ is in the "middle" of all robots in $c_i$). So, the new configuration $c_i'$ contains two robots with the same view as $r_j$, say $r_j'$ and $r_j''$. Now, execute simultaneously the two robots $r_j'$ and $r_j''$. They induce a configuration with four borders, from which no cautious algorithm can recover by Lemma~\ref{lem:morethantwoborders}, a contradiction.  
\end{proof}

Note that Lemmas~\ref{lem:morethantwoborders} and \ref{lem:noborder} justify our hypothesis that the initial configuration has exactly two borders, as those are the only solvable starting configurations by a cautious algorithm.

\begin{lemma}
\label{lem:monotonic}
A cautious algorithm that disconnects the initial visibility graph cannot solve gathering.
\end{lemma}
\begin{proof}
Since robots are not aware of $R$, the total number of robots, the algorithm must work irrespective of $R$. Suppose for the purpose of contradiction that there exists a cautious gathering algorithm that disconnects the visibility graph at some point in the execution. Without loss of generality, the visibility graph now consists of two distinct parts $\mathcal{A}$ and $\mathcal{B}$, separated in each of the two directions of the ring by at least $k>\phi$ empty nodes (That is, $H_{init}>\phi$). By Lemma~\ref{lem:morethantwoborders}, the algorithm cannot solve gathering, a contradiction.
\end{proof}

\begin{theorem}[Theorem 26 in the main part]
For $M_{init}$ even and $O_{init}$ even, 
there exists no cautious gathering algorithm with $\phi=1$, even when the initial configuration is \emph{not} edge-view-symmetric. 
\end{theorem}
\begin{proof}
Let us consider the case where $\phi=1$, 
so robots can see only colors on neighboring nodes.
Let us consider a set of two consecutive nodes $u_i$ and $u_{i+1}$ that are both occupied by robots (so, $M_{init}=O_{init}=2$). Assume that each position's color is singly colored. Obviously, the color of robots at $u_i$ must be different from the color of robots at $u_{i+1}$ (otherwise, the configuration would be edge-view-symmetric, and gathering would be impossible). Now, suppose all robots are executed synchronously; at least one of the following events happens:
\begin{enumerate}
\item Robots at $u_i$ move to $u_{i+1}$ (possibly changing colors), and robots at $u_{i+1}$ do \emph{not} move (possibly changing colors),
\item Robots at $u_{i+1}$ move to $u_i$ (possibly changing colors), and robots at $u_i$ do \emph{not} move (possibly changing colors),
\item Robots at $u_i$ move to $u_{i+1}$ (possibly changing colors), and robots at $u_{i+1}$ move to $u_i$ (possibly changing colors).
\end{enumerate}
In Case $3$, at least one of the two groups of robots \emph{must} change its color, otherwise, we obtain the same configuration, and the execution goes forever without gathering. Also, Case 3 cannot repeat forever otherwise the robots never gather. Overall, there exists a combination of colors $L_1$ and $L_2$ such that Cases 1 or 2 occurs. From this point, robots may not move anymore since the gathering algorithm is cautious. Without loss of generality, consider that robots with color $L_1$ move to the position occupied by robots with color $L_2$.

Now, consider a configuration with $M_{init}=O_{init}=6$ such that the sequence of colors is as follows: 
$L_2 L_1 L_2 L_1 L_1 L_2$. Border robots cannot move since their view is the same as in the situation with $2$ occupied nodes we presented above. Non-border robots cannot move at it would disconnect the visibility graph, which prevents gathering by a cautious algorithm by Lemma~\ref{lem:monotonic}. So, the algorithm never moves from this configuration where six positions are occupied, and hence never gathers the robots.
\end{proof}